\documentclass[12pt]{article}
\usepackage{mathrsfs,amssymb,amstext,amsmath,amsthm,eucal}
\usepackage{graphicx,subfig}
\usepackage{authblk}
\usepackage[shortlabels]{enumitem}
\usepackage[usenames,dvipsnames]{color}
\usepackage{booktabs,array}
\usepackage{hyperref}

\newcommand{\grad}{\mathrm{d}}
\newcommand{\CC}{\mathbb{C}}
\newcommand{\RR}{\mathbb{R}}
\newcommand{\ZZ}{\mathbb{Z}}
\newcommand{\SO}{\mathrm{SO}}
\newcommand{\so}{\mathfrak{so}}
\newcommand{\SU}{\mathrm{SU}}

\newcommand{\Spin}{\mathrm{Spin}}
\newcommand{\spin}{\mathfrak{spin}}
\newcommand{\dvol}{\mathrm{dvol}} 
\newcommand{\conjugate}[1]{ \overset{#1}{\sim}}

\newcommand{\Pz}{P_z}
\newcommand{\Pbz}{P_{\bar{z}}}

\newcommand{\vspan}[1]{\left\langle #1 \right\rangle}
\renewcommand{\Re}{\operatorname{Re}}
\renewcommand{\Im}{\operatorname{Im}}
\newcommand{\flatc}[1]{\underline{#1}}

\newcommand{\ef}{e^{(1)}}
\newcommand{\es}{e^{(2)}}
\newcommand{\qf}{\hat{e}^{(1)}}
\newcommand{\qs}{\hat{e}^{(2)}}
\newcommand{\cb}[2]{#1_{2 #2 -1}+i\, #1_{2 #2}}
\newcommand{\cbm}[2]{#1_{2 #2 -1}-i\, #1_{2 #2}}
\newcommand{\uw}[2]{ #1_{2 #2-1} \wedge #1_{2 #2}}
\newcommand{\efa}{e^{(1a)}}
\newcommand{\efb}{e^{(1b)}}
\newcommand{\qfa}{\hat{e}^{(1a)}}
\newcommand{\qfb}{\hat{e}^{(1b)}}
\newcommand{\er}{{e}^{(r)}}
\newcommand{\ec}{{e}^{(c)}}
\newcommand{\qr}{\hat{e}^{(r)}}
\newcommand{\qc}{\hat{e}^{(c)}}

\newtheorem*{result*}{Main Result}
\newtheorem{theorem}{Theorem}
\newtheorem{example}{Example}
\newtheorem{lemma}{Lemma}

\title{All timelike supersymmetric solutions of three-dimensional half-maximal supergravity} 
\author{Nihat Sadik Deger\footnote{sadik.deger@boun.edu.tr}}
\author{George Moutsopoulos\footnote{gmoutso@gmail.com}}
\affil{Department of Mathematics, Bogazici University, Bebek, 34342, Istanbul, Turkey}
\author{Henning Samtleben\footnote{henning.samtleben@ens-lyon.fr}}
\affil{Universit\'e de Lyon, Laboratoire de Physique, UMR 5672, CNRS et ENS de Lyon,
46 all\'ee d'Italie, F-69364 Lyon CEDEX 07, France} 
\author{\"{O}zg\"{u}r Sar{\i}o\u{g}lu\footnote{sarioglu@metu.edu.tr}}
\affil{Department of Physics,
Middle East Technical University, 06800, Ankara, Turkey} 
\date{\today}

\begin{document}
\maketitle
\begin{abstract}
We first classify all supersymmetric solutions of the 3-dimensional half-maximal 
ungauged supergravity that possess a timelike Killing vector by considering their 
identification  
under the complexification of the local symmetry of the theory. It is found that only solutions
that preserve $16/2^n, 1 \leq n \leq 3$ real supersymmetries are allowed.
We then classify supersymmetric solutions under 
the real local symmetry of the theory and we are able to solve 
the equations of motion for all of them. It is shown that all such solutions can be expressed as a direct 
sum of solutions of the integrable Liouville and $\SU(3)$ Toda systems. This completes the construction of all supersymmetric solutions of the model since the null case has already been solved.
\end{abstract}
\tableofcontents
\section{Introduction}
Supersymmetric solutions are pivotal in the study of supergravity theories since 
they possess stability properties that survive quantum deformations. 
Assuming supersymmetry renders the solution space more tractable too. 
This is because studying the first-order Killing spinor equations is easier than the second-order equations of motion.

There are various related methods of attacking the problem of finding supersymmetric solutions. In the approach that is based on spinorial geometry
one considers the reduction of the local symmetry of the theory, including the spacetime spin 
group, to the stability subgroup of Killing spinors. This method has been widely 
successful, especially so for maximally supersymmetric theories where the reduction of the spin bundle
 is straightforward 
(see for instance \cite{gillard_spinorial_2005}), but also because the method can be applied to the reduction of the generalized (hidden) 
structure group of the theory (see for instance \cite{grana_generalized_2005}). An equivalent approach is to study the various tensors formed by the Killing spinor 
bilinears as initiated by Tod in~\cite{tod_more_1995,tod_all_1983}, a method successful in various dimensions and theories (see for instance \cite{gauntlett_all_2003}). 

The latter approach was applied to study supersymmetric solutions of three-dimensional 
half-maximal supergravity in 
\cite{deger_supersymmetric_2010}. It follows from the algebra of supersymmetry variations 
that in any supergravity theory the vector formed by squaring a Killing spinor is at least Killing
which is either null or timelike. For this model the null case has been completely solved in \cite{deger_supersymmetric_2010}
and the most general solution is found to be a pp-wave. However, for the timelike case only few explicit solutions were obtained 
in \cite{deger_supersymmetric_2010}. In this paper our aim is to classify and solve for 
all supersymmetric timelike solutions of this model for which the metric is 
\begin{equation}
\grad s^2 =  \grad t^2 - e^{2\rho(x,y)}\left( \grad x^2 + \grad y^2 \right)~.
\nonumber
\end{equation}
The scalar content of the theory parametrizes the coset 
\begin{equation}
\mathcal{V} \in {G}/{K}
~,
\nonumber
\end{equation}
where we define the Lie group $G$ 
\begin{align}
\nonumber
G &= \mathrm{SO}(8,n)~,\\\intertext{its maximally compact subgroup}
K &= 
 \SO(8)\times\SO(n)\times\ZZ_2
\nonumber
 \end{align}
and their Lie algebras as $\mathfrak{g}=\so(8,n)$ and  $\mathfrak{k}=\so(8)\oplus\so(n)$, respectively. 
The coset representative is time independent, so the pull-back of the Maurer-Cartan form 
\begin{equation}
P+Q=\mathcal{V}^{-1}\grad \mathcal{V}
\nonumber
\end{equation}
only depends on the adapted coordinates $x$ and $y$. Here $P$ is the scalar current and $Q$ is the $\SO(8)\times\SO(n)$ connection.

Recently a novel classification of supersymmetric backgrounds of the three-dimensional,
maximally supersymmetric, ungauged supergravity was given in \cite{de_boer_classifying_2014}. 
The motivation there was primarily the construction of interesting supersymmetric solutions with what is termed {non-geometric monodromy}. 
Rather than fixing a Killing spinor and thus reducing the symmetry of the theory, the authors instead fixed the element $P$ under the action of some group. 
In a sense, the problem is turned on its head by asking which elements $P$ admit one Killing spinor, two Killing spinors, etc. 
The general problem of fixing $P$ this way is feasible. Moreover, the assumption of at least one supersymmetry implies that $P$ has to 
be nilpotent in some Lie algebra. More precisely, by using the Zariski topology argument, the element
\begin{equation}
\Pz = \frac12 \left( P_x - i\, P_y \right) \, \in \left( \mathfrak{g} \middle/ \mathfrak{k} \right)^{\CC}~,
\nonumber
\end{equation}
which transforms under the local group of the theory $K$, is shown to be necessarily nilpotent as an element in the 
complexified version $\mathfrak{g}^{\CC}$ of $\mathfrak{g}$, where $\mathfrak{g}$ is the Lie algebra of the global symmetry $G$. 
Note that the complexified version $\mathfrak{k}^{\CC}$ of the local algebra $\mathfrak{k}$ acts on $\Pz$ and preserves nilpotency in $\mathfrak{g}^{\CC}$. 
What is then left is to classify nilpotent orbits of $\left(\mathfrak{g}/\mathfrak{k}\right)^{\CC}$ under $K^{\CC}$, to which $\Pz$ should belong. This is particularly 
attractive as nilpotent orbits are finite and can be classified for all classical groups. For the classification one then uses the Kostant-Segikuchi 
correspondence that asserts a one-to-one correspondence of nilpotent orbits in $\left(\mathfrak{g}/\mathfrak{k}\right)^{\CC}$ under $K^{\CC}$ to 
nilpotent orbits in $\mathfrak{g}$ under $G$ \cite{de_boer_classifying_2014}. Although the method in \cite{de_boer_classifying_2014} is applied to 
maximally supersymmetric ungauged supergravity in three dimensions, where the global symmetry $G$ is $E_8$ and the local symmetry $K$ is the maximally compact 
subgroup $\SO(16)$, their topology argument applies identically to the half-maximal ungauged supergravity as well.

Note that an element $\Pz$ of a background that admits timelike supersymmetry is necessarily nilpotent in $\mathfrak{g}^{\CC}$ but 
the converse is not true. Therefore, after we obtain the nilpotent orbits in  $\left(\mathfrak{g}/\mathfrak{k}\right)^{\CC}$ under $K^{\CC}$  we need to 
check for supersymmetry. This can be done by testing the element $\Pz$ on the algebraic dilatino variation. We will show that this is sufficient as the integrability of the gravitino variation is indeed satisfied on-shell. 
The classification of nilpotent orbits under $K^{\CC}$ that admits supersymmetry is pretty concise to summarise. The supersymmetric 
orbits under $K^{\CC}$ to which such a $\Pz$ belongs correspond to the partitioning of $(8,n)$ into sums of $(2,2)$, $(2,1)$, $(1,0)$ and $(0,1)$. This decomposition 
can be thought of\footnote{For a concrete comparison, recall that a two-form in $\mathfrak{so}(n)$ under conjugation decomposes into two-forms in $\RR^2$ and $\RR$ subspaces.} as the decomposition of $\RR^{8,n}$ into orthogonal subspaces $\RR^{2,2}$, $\RR^{2,1}$, $\RR^{1,0}$ 
and $\RR^{0,1}$. The multiplicity $\mu$ of $(2,2)$ and multiplicity $\nu$ of $(2,1)$, and only these, determine the supersymmetry by the simple rule that each of them halve 
supersymmetry by a projection equation. Each class of elements, up to the action of $K^{\CC}$, corresponds to a unique partition. We call the class $N(\mu,\nu)$. That is, 
the classes are defined by
\begin{equation}
N(\mu,\nu) = \{ \Pz' : \Pz' \conjugate{K^{\CC}} \Pz \}~.
\nonumber
\end{equation}
A representative element for the class $N(\mu,\nu)$ is called a normal form. 
They are useful as 
they allow us to work with a concrete element and are pretty easy to write down. 
However, note that the group used to identify the elements $\Pz$ is the complexification $K^{\CC}$ of the symmetry of 
the theory $K$. Therefore, the orbits under $K^{\CC}$ may contain more than one, or even no solutions. For instance, a normal form 
under $K^{\CC}$ may not satisfy the equations of motion but some other representative that is $K^{\CC}$-conjugate to it might do. That is, 
it does not make sense to use the normal form in order to start solving the equations of motion because the equations of motion are not covariant under $K^{\CC}$. 
Therefore, we have to move on to classify the elements $\Pz$ under the real local symmetry of the theory $K$ in order to obtain exact solutions.
This means that for each class $N(\mu,\nu)$ and each element $\Pz'\in N(\mu,\nu)$, we need to find all the elements $\Pz$ that are distinct to $\Pz'$ under 
the action of $K$ but are identical to $\Pz'$ under $K^{\CC}$. We may call this space $N(\mu,\nu)/K$. 
The most general element $\Pz \in N(\mu,\nu)/K$ is still easy to write and are given in \eqref{eq:PzABMN}. The equations of motion and in 
particular the integrability equations 
for $P+Q=\mathcal{V}^{-1}\grad \mathcal{V}$ severely restrict the coefficients in $\Pz$. Consequently, the 
classification of the on-shell nilpotent elements that are in $N(\mu,\nu)$ should be refined into 
spaces $N(\mu,\nu_r, \nu_{c})$, where $\nu=\nu_r+\nu_{c}$. If $\Pz\in N(\mu,\nu)$ and is indeed part of a solution, then $\Pz \in N(\mu,\nu_r,\nu_{c})~.$
After this classification we analyze the field equations and integrability conditions and arrive at the following result: 
\begin{result*}
The timelike supersymmetric backgrounds of the three-di\-men\-sional, half-maximal, ungauged supergravity are locally parametrized by $\mu+\nu_r+2\nu_{c}$ meromorphic functions which are
solutions to $\mu+\nu_r$ copies of Liouville's equation and $\nu_{c}$ copies of an $\SU(3)$ Toda system.  
The $\mu$ and $\nu_r$ copies of Liouville's equation are distinguished by their contribution to the coset space connection $P+Q$ and to the spacetime curvature. 
Each $\nu$, $\nu_r$ and $\nu_{c}$ copy is responsible for halving supersymmetry once.
\end{result*}
We begin in section \ref{sec:theory} with an introduction to the theory and set up our conventions for the timelike backgrounds. In 
section \ref{sec:complex} we present the nilpotency classification. In section \ref{sec:real}, we do not yet use the equations of motion but we present the elements $\Pz$ in the classes up to 
the real symmetry. The restriction of $\Pz$ due to the equations of motion and the solutions themselves are in section \ref{sec:solutions}. We 
conclude in section 6 with some brief remarks. Most of the technical material is to be found in the appendices. In appendix \ref{app:spin} we review our 
spinorial conventions. We also give in appendix \ref{app:spin} some useful formulae for comparison with other methods in the literature.  In appendix 
\ref{app:glnc} we comment on a more direct matrix factorization of $\Pz$. Supersymmetry 
closure in the Zariski topology and construction of normal forms are explained in detail 
in appendices \ref{app:zariski} and \ref{app:constnormalforms}, respectively.

\section{Set up}\label{sec:theory}
\subsection{Theory}
Half-maximal ungauged supergravity in three dimensions is described in
the bosonic sector by a metric $g$ on a three-dimensional spin
manifold $M$ and the coset map
\begin{equation}
  \mathcal{V}:M\longrightarrow G/K ~,
\end{equation}
where the groups $G$ and $K$ are
\begin{align}
G &= \mathrm{SO}(8,n) \, , \\
K &= \mathrm{S}\left( \mathrm{O}(8)\times \mathrm{O}(n) \right)=\SO(8)\times \SO(n)\times \ZZ_2 \, ,
\end{align}
and their Lie algebras are $\mathfrak{g}$ and $\mathfrak{k}=\so(8)\oplus\so(n)$. We pull-back and split the Maurer-Cartan form on the symmetric decomposition
$\mathfrak{g}=\mathfrak{k}\oplus\mathfrak{p}$,
\begin{equation}
\mathcal{V}^{-1} \grad \mathcal{V} = Q + P \in \left( \mathfrak{k}\otimes T^{*}M \right) \oplus \left(
\mathfrak{p}\otimes T^{*}M\right)~,
\end{equation}
where $\mathfrak{p}=\mathfrak{g}/\mathfrak{k} = \RR^8\otimes \RR^n$. 
The action of the model is
\begin{equation}
S =\int \dvol_g\left( - R +  g^{\mu\nu} P_{\mu}^{Ir}P_{\nu}^{Ir}\right)~,
\end{equation}
where $\mu,\nu=0,1,2$ are spacetime indices, $I,A,\dot{A}=1,2,\ldots,
8$ are respectively the vector, chiral and anti-chiral indices for $\Spin(8)$, and $r, s=1, \ldots, n$ are $\SO(n)$ vector indices. 
Note that we use a mostly minus signature. The full theory was constructed already in \cite{marcus_three-dimensional_1983}. The gaugings of the theory 
were classified in \cite{nicolai_n8_2001}. For other gauged three-dimensional supergravities with various amounts of supersymmetry 
see~\cite{de_wit_gauged_2003,de_wit_locally_1993}.

From the action we derive the equations of motion
\begin{align}
 R_{\mu\nu} &=  P_{\mu}^{Ir}P_{\nu}^{Ir} \, ,\\
\mathcal{D}_{\mu}P^{\mu Ir} &\equiv \nabla_{\mu}P^{\mu I r} + Q_{\mu}{}^{IJ}P^{\mu J r}+ Q_{\mu}{}^{rs}P^{\mu I s} = 0~. \label{eq:divPmu}
\end{align}
The integrability of $P+Q=\mathcal{V}^{-1} \grad \mathcal{V}$ is $\grad P+\grad Q +(P+Q)\wedge (P+Q)=0$, or explicitly
\begin{align}
&\, \grad P^{Ir} + Q^{IJ} \wedge P^{Jr} + Q^{rs}\wedge P^{Is} =0 \label{eq:gradPmu} \, ,\\
R^{(Q)}{}^{IJ} \equiv & \,\grad Q^{IJ} + Q^{IK} \wedge Q^{KJ} + Q^{JK}\wedge Q^{IK} = - P^{Ir}\wedge P^{Jr} \, ,\\
R^{(Q)}{}^{rs} \equiv & \,\grad Q^{rs} + Q^{rt} \wedge Q^{ts} + Q^{st}\wedge Q^{rt} = - P^{Ir} \wedge P^{Is}~.
\end{align}

The full theory has $16$ real supersymmetries, which are locally given by $\epsilon^A_{\alpha}$ but we usually suppress the spacetime spinor index $\alpha=1,2$. With the 
gravitino $\psi_{\mu}$ and dilatino $\chi$ put to zero, a Killing spinor should satisfy
\begin{align}
 \delta\psi_{\mu} &= \mathcal{D}_{\mu}\epsilon^A = \nabla_{\mu}\epsilon^A -\frac14 Q_{\mu}^{IJ}\Gamma^{IJ}_{AB}\epsilon^B = 0 \label{eq:diffsusy}~,\\
 \delta\chi & = \gamma^{\mu}P_{\mu}^{Ir}\Gamma^I_{A\dot{A}}\epsilon^{A}
  =0~.\label{eq:dilorig}
\end{align}
We will use $\{\gamma^{a},\gamma^{b}\}=-2\eta^{ab}$, where $\eta^{ab}$ has mostly minus signature, and $\{\Gamma^{I},\Gamma^{J}\}=-2\delta^{IJ}$, so that all 
representations are real. We refer to appendix \ref{app:spin} for more details on our spinorial conventions. 

\subsection{Timelike backgrounds}
Let us define the vector 
\begin{equation}
V^{\mu}=\bar{\epsilon}^A\gamma^{\mu} \epsilon^A \, .
\end{equation}
Since the derivative $\mathcal{D}$ in the gravitino variation \eqref{eq:diffsusy} is in $\spin(1,2)\oplus \spin(8)$, 
the vector $V^{\mu}$ is easily shown to be parallel, i.e. $\nabla_{\mu} V_{\nu} = 0$. We may define the Killing spinor bilinear
\begin{equation}
F^{AB} = \bar{\epsilon}^A \epsilon^B = -F^{BA}~, \label{eq:FAB}
\end{equation}
in order to derive via the Fierz identity 
\begin{equation}
\epsilon^A \bar{\epsilon}^B = -\frac12 \bar{\epsilon}^B \gamma^{\mu} \epsilon^A \gamma_{\mu} +\frac12  \bar{\epsilon}^B \epsilon^A \, , 
\end{equation}
which shows that $V$ is either null or timelike:
\begin{equation}
V^{\mu}V_{\mu} = F^{AB}F^{AB} \geq 0~.
\end{equation}
The null case was completely solved and few explicit solutions for the timelike case were obtained in \cite{deger_supersymmetric_2010}.
In this paper we only consider the timelike case and so $V^{\mu}$ is a timelike covariantly constant vector. It follows that 
we can find adapted coordinates $(t,x,y)$ so that $V=\partial_t$ and the metric is 
\begin{equation}
\grad s^2 =  \grad t^2 - e^{2\rho(x,y)}\left( \grad x^2 + \grad y^2 \right)~.
\label{eq:ultrastatic}
\end{equation}
It is shown in \cite{deger_supersymmetric_2010} that $\partial_t$ also leaves the coset representative invariant up to a local $K$ transformation, so 
in particular we may choose a gauge where $Q_t=P_t = 0$. 

The Einstein equations of motion for the metric \eqref{eq:ultrastatic} are only non-trivial in the $(x,y)$ components,
\begin{equation}
 g_{ij} \, e^{-2\rho}\partial_k \partial_k \rho = P^{Ir}_i P^{Ir}_j~,\quad i,j=1,2.
\end{equation}
It thus follows that
\begin{align}
P_x^{Ir} P_y^{Ir} &=0 \, , \\
P_x^{Ir} P_x^{Ir} &=  P_y^{Ir} P_y^{Ir} = - \partial_i \partial_i \rho~.
\end{align}
If we then define $z=x+iy$ and  
\begin{equation}
\Pz^{Ir} \equiv \frac12 \left( P_x^{Ir} - i P_y^{Ir} \right)~,
\label{eq:defPz}
\end{equation}
the non-trivial components of the Einstein's equation are:
\begin{align}
\Pz^{Ir} \Pz^{Ir} &= 0 \label{eq:EinND} \, ,\\
\Pbz^{Ir}\Pz^{Ir} &= -2 \partial_z\partial_{\bar{z}}\rho~. \label{eq:PPz}
\end{align}
Equation \eqref{eq:PPz} is the only equation involving the conformal factor in our formalism.

We now turn to the equation of motion and integrability equation for $\Pz^{Ir}$, \eqref{eq:divPmu} and \eqref{eq:gradPmu}. They respectively become
\begin{align}
\Re\left( \partial_z \Pbz^{Ir} + Q_z^{IJ} \Pbz^{Jr} + Q_z^{rs} \Pbz^{Is} \right) & = 0 \label{eq:divP} \, , \\
\Im\left( \partial_z \Pbz^{Ir} + Q_z^{IJ} \Pbz^{Jr} + Q_z^{rs} \Pbz^{Is} \right) & = 0 \label{eq:dP}
~.\end{align}
Combining them, the equation of motion for $\Pz^{Ir}$ is
\begin{equation}
\mathcal{D}_{\bar{z}}\Pz^{Ir} \equiv \partial_{\bar{z}} P_z^{Ir} +  Q_{\bar{z}}^{IJ} \Pz^{Jr} + Q_{\bar{z}}^{rs} \Pz^{Is} = 0~. \label{eq:DPz}
\end{equation}
Finally, the two integrability equations for $Q^{IJ}$ and $Q^{rs}$ are written as
\begin{align}
\Im \left( \mathcal{D}_{\bar{z}} Q_z^{IJ} \right) 
&= 
\Im \left( \partial_{\bar{z}} Q_z^{IJ} + Q_{\bar{z}}^{IK} Q_z^{KJ} + Q_{\bar{z}}^{JK} Q_z^{IK} \right)= - \Im\left( \Pbz^{Ir} \Pz^{Jr} \right) ~, \label{eq:DQIz}\\
\Im \left( \mathcal{D}_{\bar{z}} Q_z^{rs} \right) 
&= 
\Im \left( \partial_{\bar{z}} Q_z^{rs} + Q_{\bar{z}}^{rt} Q_z^{ts} + Q_{\bar{z}}^{st} Q_z^{rt} \right)= - \Im\left( \Pbz^{Ir} \Pz^{Is} \right) ~. \label{eq:DQrz}
\end{align}
The full set of equations of motion, including the coset integrability equations, are \eqref{eq:EinND}, \eqref{eq:PPz}, \eqref{eq:DPz}, 
\eqref{eq:DQIz} and \eqref{eq:DQrz}. Only \eqref{eq:PPz} involves the conformal factor $e^{2\rho}$ and we can solve the latter three independent of the first two.
Now we will analyze them assuming that the solution preserves some supersymmetry.

\subsection{Timelike Killing spinors}
Let us define the complex $\Spin(8)$ spinor
\begin{align}
\epsilon_z^A \equiv \epsilon_1^A + i \epsilon_2^A~,
\end{align}
which under a rotation in the $(x,y)$ plane has weight $-1/2$, see also 
appendix \ref{app:spin}. The dilatino Killing spinor equation \eqref{eq:dilorig} becomes
\begin{equation}
  P^{Ir}_z\Gamma_{A\dot{A}}^I \epsilon^A_{\bar{z}}=0~, \label{eq:algsusy}
\end{equation}
where $P^{Ir}_z$ was defined in \eqref{eq:defPz}. We will first show that the gravitino Killing spinor equation \eqref{eq:diffsusy} is integrable
provided that the equations of motion  and the dilatino variation \eqref{eq:algsusy} hold. Note that the $t$-component of the 
equation \eqref{eq:diffsusy} is simply $\partial_t\epsilon^A_z=0$, whence Killing spinors are time-independent. 
The curvature of the supersymmetric connection \eqref{eq:diffsusy} should stabilize a Killing spinor,
\begin{equation}
\left( -\frac14 R_{\mu\nu ab} \gamma^{ab} \delta_{AB} - \frac14 R^{(Q)}_{\mu\nu}{}^{IJ}\Gamma^{IJ}_{AB} \right) \epsilon^B = 0~, \label{eq:integKS}
\end{equation}
a condition with non-vanishing components only for $\mu,\nu=i,j$. In particular, the only non-trivial 
Riemann curvature tensor component is $R_{{1}{2}{1}{2}}= e^{2\rho}\partial_i\partial_i\rho$. The integrability equation for Killing spinors \eqref{eq:integKS} is directly 
equivalent to
\begin{equation}
- 2 i \partial_z \partial_{\bar{z}} \rho \,\epsilon^A_z + \Im \left( \mathcal{D}_{\bar{z}} Q_{{z}}^{IJ} \right)\Gamma^{IJ}_{AB} \epsilon^B_z = 0~.
\end{equation}
However, combining the Einstein equation \eqref{eq:PPz} and the coset integrability equation \eqref{eq:DQIz}, we may show that the curvature of the supersymmetry 
connection (the operator acting on $\epsilon^A_z$ in \eqref{eq:integKS}) is identically zero:
\begin{equation}
\Gamma^I_{A\dot{A}}\Gamma^J_{B\dot{A}}\left( -2 i\,  \partial_z \partial_{\bar{z}} \rho \,\delta^{IJ} + \Im \left( \mathcal{D}_{\bar{z}} Q_{{z}}^{IJ} \right) \right) = 0~.
\end{equation}
The algebraic equation \eqref{eq:algsusy} is therefore a necessary and sufficient condition for the existence of Killing spinors.

One may also show that the $zz$-component of the Einstein equation \eqref{eq:PPz}
is redundant. Indeed, multiplying \eqref{eq:algsusy} with  $
\Pz^{Js}\Gamma^J_{B\dot{A}}=0$ and symmetrizing over $(r,s)$ one arrives at
\begin{equation}
\Pz^{Ir}\Pz^{Ir} \epsilon_z^A=0~,
\end{equation} 
which for a non-zero spinor gives precisely  $\Pz^{Ir}\Pz^{Ir}=0$. Timelike supersymmetric solutions are thus entirely described by the coset equations 
\eqref{eq:DQIz}, \eqref{eq:DQrz} 
and \eqref{eq:DPz} that determine $P$ and $Q$, the Einstein equation \eqref{eq:PPz} 
that determines $\rho$, and finally the condition that $\Pz$ admits Killing spinors via the algebraic equation \eqref{eq:algsusy}.
Therefore, when the equations of motion are satisfied, Killing spinors are characterized only by \eqref{eq:algsusy}. Note that if $\epsilon_z^A$ is a Killing spinor, 
then so is $i\,\epsilon_z^A$. We may thus assert the following:
\begin{theorem}
\label{thm:evensusy}
Supersymmetric solutions with a timelike Killing vector admit an even amount of real supersymmetry and form a complex vector space.
\end{theorem}
We will see in Theorem \ref{thm:BPS}
that not only is the amount of supersymmetry even, but it comes in powers of two: $16, 8, 4, 2$.

\section{Nilpotency}\label{sec:complex}
Our strategy in this section is to set aside the equations of motion for $P$, $Q$ and $\rho$, and classify instead all elements $\Pz$ that admit 
supersymmetry via equation \eqref{eq:algsusy}. The classification is with respect to $K^{\CC}$, the complexification of the local symmetry of the theory. 
That is, we identify all admissible $\Pz$ up to the action of $K^{\CC}$. The classes are parametrized by integers $\mu$ and $\nu$ and we call each class $N(\mu,\nu)$. 

\subsection{Proof of nilpotency}
We note that the symmetry of the dilatino supersymmetry equation \eqref{eq:algsusy} is $\mathrm{SO}(8)^{\CC}\times \mathrm{GL}(n,\CC)$. Indeed, $\mathrm{SO}(8)^{\CC}$ is 
the group that preserves the gamma
matrices of the 8-dimensional Clifford algebra. For instance, take $m\in
\mathfrak{so}(8)$ and note that since\footnote{$m_{AB}= -\frac14
  m_{IJ}\Gamma^{IJ}_{AB}$ and similarly for $m_{\dot{A}\dot{B}}$.}
\begin{equation}
  m^I{}_J \Gamma^{J}_{A\dot{A}} = \Gamma^I_{B\dot{A}} m^{B}{}_A +
  \Gamma^I_{A\dot{B}} m^{\dot{B}}{}_{\dot{A}}~,
\end{equation}
and all representations are real, we can complexify the Lie algebra element $m$. On the other hand, the index $r$  in \eqref{eq:algsusy} is a free index, 
whence the symmetry $\mathrm{GL}(n,\CC)$. 

Classifying $\Pz$ up to the action of $\mathrm{SO}(8)^{\CC}\times \mathrm{GL}(n,\CC)$ turns out to be too strong. However, it does prove that the algebraic supersymmetry 
equation \eqref{eq:algsusy} is a set of projection equations that halve the real supersymmetries according to $16, 8, 4, 2$, we give the proof in appendix 
\ref{app:glnc}. Instead, we classify the elements $\Pz$ up to the action of
\begin{equation}
  K^{\CC} = \mathrm{SO}(8)^{\CC}\times \mathrm{SO}(n)^{\CC} \times \ZZ_2~.
\end{equation}
Since it is a symmetry of the algebraic supersymmetry equation, we may consider the orbit space of the $(\mathfrak{g}/\mathfrak{k})^{\CC}$ where 
$\Pz$ belongs to, up to the action of $K^{\CC}:(\mathfrak{g}/\mathfrak{k})^{\CC} \rightarrow (\mathfrak{g}/\mathfrak{k})^{\CC}$. That is, since any 
other element in the same orbit admits the same amount of supersymmetry we may consider the orbit as a whole. The group $K^{\CC}$ is not a symmetry of the theory, in 
contrast to the group $K$, but one may hope to move from this classification to orbits under $K$ once the first are obtained, which we do in section \ref{sec:real}. Note
also that relative 
to $\mathrm{SO}(8)^{\CC}\times \mathrm{GL}(n,\CC)$, the orbit space of the action $K^{\CC}$ is more 
fine grained and thus perhaps more useful. In fact, it turns out that the orbit space under $\mathrm{SO}(8)^{\CC}\times \mathrm{GL}(n,\CC)$ is 
labeled by the amount of supersymmetry. Similar to the case of maximal supergravity \cite{de_boer_classifying_2014}, we will now show that 
the element $\Pz \in  \mathfrak{p}^{\CC}=(\mathfrak{g}/\mathfrak{k})^{\CC}$ is nilpotent in the adjoint representation of $\mathfrak{g}^{\CC}$. That is, with 
the symmetric decomposition
\begin{equation}
\mathfrak{g}^{\CC} = \so(8,n)^{\CC} = \mathfrak{k}^{\CC} \oplus  \mathfrak{p}^{\CC}~,
\end{equation}
we will now show that $\left(\textrm{ad}_{\Pz}\right)^{p+1}=0$ for some positive integer $p$. 

Our proof closely follows \cite{de_boer_classifying_2014}. Consider an element $\Pz \in \mathfrak{p}^{\CC} $. The
Jordan-Chevalley decomposition tells us that it can be written as a sum of a
semi-simple element and a nilpotent element
\begin{equation}
\label{eq:Jordan1}
  \Pz = P_S + P_N~,
\end{equation}
with $P_S,P_N \in \mathfrak{p}^{\CC} \subset \so(8,n)^{\CC}$ and
$[P_S,P_N]=0$, see proposition 3 in \cite{kostant_orbits_1971}. Assume that \eqref{eq:Jordan1} admits $\tilde{n}>0$ algebraic Killing supersymmetries according to 
\eqref{eq:algsusy}. Consider
then the orbit $O$ of $\Pz$ under $K^{\CC}$ and assume $P_S \neq 0$. 
The algebraic supersymmetry equation \eqref{eq:algsusy} implies that elements in the closure $\bar{O}$ of the orbit $O$ in the Zariski topology preserve at least $\tilde{n}$
supersymmetries, a result of \cite{de_boer_classifying_2014} that we review\footnote{The reader may consult the definition of the Zariski topology also in the appendix \ref{app:zariski}.} in appendix \ref{app:zariski}. At the same time, it can be shown that in the 
Jordan-Chevalley decomposition, the semi-simple element $P_{S}$ is in the closure of the orbit, $P_S \in \bar{O}$,  see lemma 11 in \cite{kostant_orbits_1971}. 
Furthermore, any semi-simple element $P_S$ in $\mathfrak{p}^{\CC}$ is $K^{\CC}$-conjugate to an element in the Cartan subalgebra 
in $\mathfrak{p}^{\CC}$, by virtue of its semi-simplicity alone. In summary, if $\Pz$ preserves $\tilde{n}$ supersymmetries and $P_S \neq0$, then there is an element in 
the Cartan subalgebra in $\mathfrak{p}^{\CC}$ that preserves at least $\tilde{n}$ supersymmetries. Yet, it is easy to show that an element in the Cartan subalgebra in
$\mathfrak{p}^{\CC}$ does not preserve any supersymmetry and hence $P_S$ has to vanish. In order to show this, assume first an orthonormal basis $e_I$ of $\RR^8$ and an 
orthonormal 
basis $\hat{e}_r$ of $\RR^n$. Then an element in the Cartan subalgebra in $\mathfrak{p}^{\CC}$ has to be diagonal and is expanded in this basis as
\begin{equation}
  P_S^{Ir} e_I \otimes \hat{e}_r = P_S^{11} e_1 \otimes \hat{e}_1 + P_S^{22} e_2 \otimes  \hat{e}_2+ \cdots ~.
\end{equation}
The algebraic supersymmetry equation \eqref{eq:algsusy} for $r=1$ becomes (if the component $P_S^{11}$ is zero, take instead the first non-zero element)
\begin{equation}
  \Gamma^1_{A\dot{A}}\epsilon^A_z = 0~.
\end{equation}
Since the gamma matrix $\Gamma^1$ squares to $-1$, this equation cannot admit a non-zero solution for $\epsilon^A_z$. Hence, if $\Pz$ admits
some supersymmetry then $\Pz=P_N$ and the orbit $O$ of $\Pz$ under $K^{\CC}$ is
nilpotent.

Our task is then to classify the nilpotent orbits in
$\mathfrak{p}^{\CC}$ under $K^{\CC}$, a space we may write as
\begin{equation}
\left.  \text{Nil}[\mathfrak{p}^{\CC}] \middle/ K^{\CC} \right. ~.
\end{equation}
To this aid, we use the Kostant-Segikuchi correspondence, which is a correspondence between nilpotent
elements in $\mathfrak{g}$ up to the action of $G$ and nilpotent elements in $\mathfrak{p}^{\CC}$ up to the action of $K^{\CC}$:
\begin{equation}
 \text{Nil}[ \mathfrak{g} ] / G  = \text{Nil}[ \mathfrak{p}^{\CC}]/ K^{\CC}~.
\end{equation}
For more details, see appendix (\ref{sec:KSC}).

\subsection{Indecomposable types and their normal forms}
\label{sec:indecomposable}

Our goal now is to classify nilpotent elements of $\mathfrak{so}(m,n)$ up to conjugacy by $O(m,n)$. In particular we will construct normal forms, which are representatives in each class. We begin by developing the notions of decomposable and indecomposable types of elements in the Lie algebra. Note that this will not be the same as the notion of a module's decomposition into indecomposable submodules, one should rather think here of a block diagonal form of a matrix. Consider for example a two-form in $\mathfrak{so}(n)$ up to the action of $SO(n)$. We know that one can decompose it in some orthonormal basis into a block diagonal form of antisymmetric $2\times 2$ matrices, each proportional to the same antisymmetric real Pauli matrix, and trailing zeros. In this case, the $2\times 2$ antisymmetric matrices are indecomposable, that is to say they cannot be decomposed into smaller block diagonal forms. We wish to do the equivalent for the nilpotent elements in $\mathfrak{so}(8,n)$. Normal forms for elements in the classical linear groups have been described but not explicitly written in   
\cite{burgoyne_conjugacy_1977} (see also \cite{djokovic_normal_1983}).  

Consider the Lie algebra $L(V,\tau,\sigma)$ of a linear group that acts on a complex vector space $V$, preserves the bilinear $\tau$ and is compatible with the real or 
pseudoreal structure $\sigma$, where the latter is compatible with $\tau$\footnote{For our problem the group is $\mathrm{O}(m,n)$, $V=\mathbb{C}^{m+n}$, $\tau$ is 
symmetric and 
$\sigma$ is a real structure.}. Let $A\in L(V,\tau,\sigma)$ and $A'\in L(V',\tau',\sigma')$. We 
take $(A,V)$ and $(A', V')$ as equivalent if 
there is an isomorphism $\phi$ such that:
\begin{subequations}
\begin{align}
 \phi: V & \rightarrow V' \, ,\\
 \phi A & = A' \phi \, ,\\
 \phi \sigma & = \sigma' \phi \, ,\\
 \tau(\phi (\cdot) , \phi (\cdot) ) &= \tau' (\cdot,\cdot)~.
\end{align}
\end{subequations}
The equivalence class defines a so-called type $\Delta$, that is $(A,V)\in \Delta$. 

If $(A,V)\in \Delta$ and $A$ is reducible on the direct sum of $\tau$-orthogonal, $\sigma$-invariant subspaces $V=V_1 \oplus V_2$, that is $A V_1 \subset V_1$ and $A V_2 \subset V_2$, then note that 
$L(V_i,\tau|,\sigma|)$ is well-defined and we can write
$A \in L(V_i,\tau,\sigma)$ and $(A,V_i)\in \Delta_i$ for a type in the restricted linear algebra. In this case, we 
define the decomposition of types
\begin{equation}
\label{eq:TypeDecompDef}
\Delta = \Delta_1  \oplus \Delta_2~.
\end{equation}
Note that we also have
\begin{equation}
\label{eq:dimaddsup}
\dim \Delta = \dim\Delta_1+\dim\Delta_2 \, ,
\end{equation}
for the dimensions of the corresponding vector space decomposition. For the case of symmetric $\tau$, the signature of the two types $\Delta_1$ and $\Delta_2$ 
should also add up to that of $\Delta$, a property that we will use in our classification. The notion of decomposition of types in \eqref{eq:TypeDecompDef} lends 
to the definition of an \emph{indecomposable type}. That is, an indecomposable type $\Delta$ is such that it cannot be decomposed as in \eqref{eq:TypeDecompDef}. Finally, the decomposition of the type $\Delta$ into indecomposable types $\Delta_i$,
\begin{equation}
\Delta = \oplus_i \Delta_i \, ,
\label{eq:TypeDecompIndec}
\end{equation}
can be shown to be essentially unique.

\begin{table}
\begin{center}
\begin{tabular}{>{$}l<{$}>{$}l<{$}>{$}l<{$}>{$}l<{$}}
  \text{type}&
  \text{condition}&
  \text{condition}&
  \text{signature} \\
\toprule
  \Delta_p(\zeta,-\zeta,\bar{\zeta},-\bar{\zeta})& 
  \zeta\neq \pm \bar{\zeta}&
  p\in\mathbb{N}&
  (2(1+p),2(1+p))\\
  \Delta_p(\zeta,-\zeta)&
  \zeta\in\mathbb{R}^{*}&
  p\in\mathbb{N}&
  (1+p,1+p)
  \\
  \Delta_p^\pm(\zeta,-\zeta) &
  \zeta\in i\mathbb{R}^* &
  p\in 2\mathbb{N}&
  \pm(-1)^{p/2}(p+2,p)\\
  \Delta_p^\pm(\zeta,-\zeta) &
  \zeta\in i\mathbb{R}^* &
  p\in 2\mathbb{N}+1&
  (p+1,p+1)\\
  \Delta^\pm_p(0)&
  -&
  p\in 2\mathbb{N}&
  \pm(-1)^{p/2}(\frac{p}{2} + 1,\frac{p}{2})\\
  \Delta_p(0,0)&
  -&
  p\in 2\mathbb{N}+1&
  ({p} + 1,{p}+1)
\end{tabular}
\caption{Indecomposable types of $\mathrm{O}(m,n)$, where the negative sign in the signature means: $-(s_1,s_2) \equiv (s_2,s_1)$.}
\label{tab:IndecOmn}
\end{center}
\end{table}

We give the
indecomposable types $\Delta$ of $\mathfrak{so}(m,n)$ in table \ref{tab:IndecOmn}. The types in
table \ref{tab:IndecOmn} are denoted by $\Delta_p(\zeta,\cdots)$,
where $p$ is the order of its nilpotent part $N$ in the fundamental and in parentheses the $(\zeta,\cdots)$ are the eigenvalues
of its semisimple part $S$. We also list the dimension and signature that any given type belongs to. Under 
a decomposition into indecomposables, see \eqref{eq:TypeDecompIndec}, the signatures add up as in \eqref{eq:dimaddsup}. An 
algorithm to find the types of elements in $\mathfrak{so}(m,n)$ is to partition the signature $(m,n)$ into numbers $(m_i,n_i)$ that correspond 
to the indecomposable types in table \ref{tab:IndecOmn}.

\begin{example}
\label{example1}
A nilpotent element in $\mathfrak{so}(2,2)$ can be decomposed into
indecomposables of signature $(2,2)$, $(1,0)$, $(0,1)$, $(2,1)$ and
$(1,2)$. These correspond, respectively, to the indecomposable types $\Delta_1(0,0)$,
$\Delta_0^+(0)$, $\Delta_0^-(0)$, $\Delta_2^{-}(0)$ and $\Delta_2^+(0)$. The possible partitions are found by matching up the signature. We thus get the 
following types of nilpotent elements in $\mathfrak{so}(2,2)$
\[\begin{aligned}
\Delta_1(0,0), \quad \Delta_2^-(0)+\Delta_0^-(0), &\\
\Delta_2^+(0)+\Delta_0^+(0), \quad 2\Delta_0^+(0)+ 2\Delta_0^-(0) \, .&
\end{aligned}\]
Each nilpotent element of $\mathfrak{so}(2,2)$ is $\mathrm{O}(2,2)$-conjugate to exactly one of these four types.
\end{example}

From the table we see that if the indecomposable type is nilpotent, then there are only two possibilities: type $\Delta_p^{\pm}(0)$ and type $\Delta_p(0,0)$. We 
construct normal forms 
for these types in appendix \ref{app:normalforms}, and in \ref{app:KST} we give  their corresponding Kostant-Segikuchi triples in $\so(m,n)$. Via the 
Kostant-Segikuchi correspondence, we thus arrive at the normal forms for the indecomposable nilpotent elements in $\mathfrak{p}^{\CC}$ up to the action of 
$K^{\CC}$ which we give in appendix \ref{app:norminpc}.



\subsection{Supersymmetric nilpotency}
\label{sec:susynil}
In the previous subsection, we classified the complex nilpotent elements that $\Pz$ necessarily belongs to. However, not all of them admit supersymmetry.
We need to select those that admit a non-zero amount of supersymmetry according 
to the algebraic supersymmetry equation which leads us to:
\begin{theorem}\label{thm:susynil}
Assume that $\Pz$ admits some supersymmetry. If we decompose
the element $\Pz$ into nilpotent indecomposable types of
$\mathrm{SO}(8,n)$, then the following hold
\begin{enumerate}[a)]
\item Type $\Delta_p(0,0)$ for $p\geq 3$ does not appear in the decomposition,
\item The multiplicity $\mu$ of $\Delta_1(0,0)$ is responsible for
  projecting supersymmetry to a fraction $(1/2)^{\mu}$,
\item Type $\Delta_p(0)$ for $p\geq 4$ does not appear in the
  decomposition,
\item The multiplicity of $\Delta_0(0)$ in the decomposition does not
  affect supersymmetry,
\item Type $\Delta_2^+(0)$ does not appear in the decomposition. The multiplicity $\nu$  of $\Delta_2^{-}(0)$  is responsible for
  projecting supersymmetry to a fraction $(1/2)^{\nu}$.
\end{enumerate}
\end{theorem}

We give the proof of theorem \ref{thm:susynil} in appendix \ref{app:susynil}. 
Types $\Delta_2^{-}(0)$ and $\Delta_1(0,0)$ are the only ones 
that determine supersymmetry, because type $\Delta_0^{\pm}(0)$ is represented by 
$\Pz=0$.  Assume a basis $e_I$ of $\RR^8$ and $\hat{e}_r$ of $\RR^n$ related to 
$\SO(8,n)$ ungauged supergravity. Normal forms corresponding to each 
indecomposable type  
can be written in tensor product form (see \eqref{eq:eissusy1} and  
\eqref{eq:eisdelta1} with some relabeling):
\begin{align}
\label{eq:epure1}
\left( \pm e_1 +i\, e_2 \right)\otimes \hat{e}_1 \,  & \in 
\Delta_2^{-}(0) \, ,\\
\label{eq:epure2}
\left(\pm e_1 + i \,e_2 \right) \otimes 
\left( \hat{e}_1 \pm i\,  \hat{e}_2 \right) \, & \in \Delta_1(0,0)~.
\end{align}
If $\Pz$ in its 
decomposition into  $\Delta_2^{-}(0)$ and $\Delta_1(0,0)$ does not span the 
whole space $\mathfrak{p}^{\CC}$, 
then one can use a parity transformation in the perpendicular directions and absorb the signs that appear in \eqref{eq:epure1} and  \eqref{eq:epure2}. 
If on the other hand the element $\Pz$ spans the whole space, then all signs are again absorbed because the sign of the last type that appears in the decomposition 
is fixed to be one because of the chirality of Killing spinors. Indeed, the algebraic supersymmetry equation for each type in \eqref{eq:epure1} 
and  \eqref{eq:epure2} is manifestly that of a BPS projection equation
\begin{equation}
\left(\Gamma^{1}_{A\dot{A}} +i\ \Gamma^{2}_{A\dot{A}}\right) \epsilon^A_z = 
0~,
\end{equation}
where $\Gamma^1$ corresponds to ${e}_1$ in \eqref{eq:epure1} or 
\eqref{eq:epure2} and $\Gamma^{2}$ corresponds to 
$e_2$ in  \eqref{eq:epure1} or \eqref{eq:epure2}. Note that the $\epsilon^A_z$ 
appearing in this equation is 
only $K^{\CC}$-conjugate to the 
actual supergravity Killing spinor.

At this point we introduce the following notation: A supersymmetric element 
$\Pz$ is said to belong to type $N(\mu,\nu)$ if it decomposes into types as
\begin{equation}
\Pz \in \left( \underbrace{\Delta_1(0,0) \oplus \cdots \oplus \Delta_1(0,0)}_{\mu\text{ times }}\right) \oplus
 \left( \underbrace{\Delta_2^{-}(0) \oplus \cdots \oplus 
\Delta_2^-(0)}_{\nu\text{ times }} \right)~.
\end{equation}
By using \eqref{eq:epure1} 
and \eqref{eq:epure2}, a supersymmetric element 
$\Pz \in N(\mu,\nu)$ is $K^{\CC}$-conjugate to
\begin{multline} \label{eq:Pzatomic}
 \Pz \conjugate{K^{\CC}} \underbrace{\left( e_1 + i \, e_2 \right) \otimes 
\left( \hat{e}_1 + i\, \hat{e}_2 \right) + \cdots  + \left( e_{2\mu-1} + i \, 
e_{2\mu}\right) \otimes \left( \hat{e}_{2\mu-1} + i\, \hat{e}_{2\mu} 
\right)}_{\mu\text{ terms}} \\ +
 \underbrace{\left( e_{2\mu+1} + i \, e_{2\mu+2} \right) \otimes  
\hat{e}_{2\mu+1} + \cdots + \left( e_{2\mu+2\nu-1} + i \, e_{2\mu+2\nu}\right) 
\otimes  
 \hat{e}_{2\mu+\nu}}_{\nu\text{ terms}} \,.
\end{multline}
It follows from the signatures of the types in table \ref{tab:IndecOmn} that 
each class $N(\mu,\nu)$ corresponds to the partition of
$(8,n)$ into the sums of $(1,0)$, $(0,1)$, $(2,1)$ and $(2,2)$, by using the 
convention
\begin{equation}
(a_1,b_1)+(a_2,b_2)=(a_1+a_2,b_1+b_2)~,
\end{equation}
with multiplicity $\mu$ of $(2,2)$ and multiplicity $\nu$ of
$(2,1)$.

\begin{example}
 A supersymmetric 
element $\Pz$ of type $N(2,1)$ is $K^{\CC}$-conjugate to
\begin{equation}
\Pz \conjugate{K^{\CC}} \underbrace{\left( e_1 + i\, e_2 \right) \otimes \left( \hat{e}_1 + i \, \hat{e}_2 \right) +  \left( e_3 + i\, e_4 \right) \otimes \left( \hat{e}_3 + i \, \hat{e}_4 \right)}_{\mu=2\text{ terms}} +  \underbrace{\left( e_5 + i\, e_6 \right) \otimes  \hat{e}_5 }_{\nu=1\text{ terms}}~.
\end{equation}
With $\Pz=\Pz^{Ir} e_I \otimes \hat{e}_r$ and taking the components $r=1,3,5$, the algebraic supersymmetry equation (\ref{eq:algsusy})
is $K^{\CC}$-invariant and becomes
\begin{align}
\left(\Gamma^{1}_{A\dot{A}} + i \Gamma^2_{A\dot{A}} \right)\epsilon^{A}_{z} &= 0  \, ,\\
\left(\Gamma^{3}_{A\dot{A}} + i \Gamma^4_{A\dot{A}} \right)\epsilon^{A}_{z} &= 0 \, ,\\
\left(\Gamma^{5}_{A\dot{A}} + i \Gamma^6_{A\dot{A}} \right)\epsilon^{A}_{z} &= 0 ~.
\end{align}
The matrices $i\Gamma^{12}$, $i\Gamma^{34}$ and $i\Gamma^{56}$ are 
compatible projection operators such that $\Gamma^{I_1 
\ldots I_{2k}}_{AA}=0$ for $k\neq 0,4$. They therefore halve real supersymmetry 
down to $16/2^3=2$. 
\end{example}

By generalizing the above example, we reach 
\begin{theorem}\label{thm:BPS}
The real supersymmetries of a timelike supersymmetric background in ungauged half-maximal supergravity comes 
in powers of $2$, that is $16,8,4,2$. In particular, class $N(\mu,\nu)$ 
has $16/2^{\mu+\nu}$ real supersymmetries for $\mu+\nu<4$ and $2$ real 
supersymmetries for $\mu+\nu=4$. 
\end{theorem}
Note that having only one real supersymmetry is 
excluded because according to theorem \ref{thm:evensusy} the vector space of 
Killing spinors 
is complex. In the case of $\mu+\nu=4$ there are only three independent 
BPS projections due to chirality. Theorem \ref{thm:BPS} can also be shown in a 
more direct approach, which we do in appendix \ref{app:glnc}. 

\section{Identification under $K$}\label{sec:real}
The classification under $K^{\CC}$ is genuine and powerful. However, it is of 
little use if we cannot access the solutions. If $P_1\in N(\mu,\nu)$ is a normal 
form in the class but is not part of a solution, it does not follow that a 
conjugate element 
$P_2 \conjugate{K^{\CC}} P_1$ is also not a solution. If $P_1$ is indeed a solution, by using only $P_1$ we miss all other potential solutions that are 
related to $P_1$ by $K^{\CC}$ but not related to it by $K$, where the latter is the actual symmetry of the theory.
Therefore, we should move from normal forms of a class $N(\mu,\nu)$, that is under the identification of
\begin{equation}
K^{\CC} = \left( \mathrm{SO}(8)\times \mathrm{SO}(n) \right)^{\CC} \times \ZZ_2~,
\end{equation}
to all elements that are identified under
\begin{equation}
K =  \mathrm{SO}(8)\times \mathrm{SO}(n) \times \ZZ_2~.
\end{equation}
We may think of starting with a specific normal form $P_1\in N(\mu,\nu)$ and act on it with all possible $K^{\CC}$ rotations, modulo 
its stabilizer that leaves the normal form invariant anyway, thus obtaining all elements in $N(\mu,\nu)$. Subsequently, we should identify under $K$ and obtain 
the space that we call $N(\mu,\nu)/K$. We are thus interested in the double quotient on the right-hand side of
\begin{equation} \label{eq:doublecoset}
N(\mu,\nu)/K = \left( \mathrm{SO}(8)\times \mathrm{SO}(n)\right)  \backslash  \left( \mathrm{SO}(8)\times \mathrm{SO}(n)\right)^{\CC} / \text{Stab}(N(\mu,\nu))~.
\end{equation}
Note that the normal forms in \eqref{eq:Pzatomic} do not contain any 
coefficients so the spacetime variance of $\Pz$ comes from the double coset 
alone.

We will not parametrize the double quotient \eqref{eq:doublecoset} directly. 
Instead, we will use the action  of a real orthogonal group on the 
complexification of its associated vector space, which we describe in the next 
subsection. Then, we will be able to write the most general form of a 
$\Pz\in N(\mu,\nu)$ after identifying the elements up to the real local symmetry 
of the theory. 

\subsection{Complex vectors}
\label{sec:complexvectors}
We begin with the action of $\mathrm{O}(m)$ on complex vectors in $\CC^m$ with inner product defined as $A \cdot B=\sum_I A_IB_I$. We will later specialize 
for $m=8$ and $m=n$. This subsection will eventually serve our goal to fix $\Pz \in \CC^8 \otimes \CC^m$ under the action of $\SO(8)\times \SO(m)\times \ZZ_2$.

Let us first consider complex null vectors, for instance a vector $v\in \CC^m$ such that $v \cdot v = 0$. Let us use an orthonormal basis $\{e_I\}$ of $\CC^m$. It 
is clear that one may $\mathrm{O}(m)$-rotate the real part of $v$ to only have a component in $e_1$ and then rotate its imaginary part, by using the stabilizer 
$\mathrm{O}(m-1)$, to have components in $e_1$ and $e_2$. The condition $v \cdot v=0$ though implies that its expansion in components is
\begin{equation} \label{eq:firstnull}
v = v^1( e_1 + i\,  e_2) \, ,
\end{equation}
in terms of some real $v^1$ that can be chosen positive.
If we wish to fix $v$ under the action of $\mathrm{O}(m)^{\CC}$ instead, there is a hyperbolic element in $\SO(2)^{\CC}\subset \mathrm{SO}(m)^{\CC}$ that scales $v$ and 
so $v^1$ can be set to one. Assume now an ordered  set of $\mu$ complex null vectors $\{v_{(i)}\}_{i=1}^{\mu}$ that are linearly independent and orthogonal to each other. 
We may fix the first vector $v_{(1)}$ as in \eqref{eq:firstnull}, fix the second vector $v_{(2)}$ to only have components in $\vspan{e_1,e_2,e_3,e_4}$, etc., a modification 
of the QR decomposition. Since the vectors are orthogonal to each other, only half of their coefficients are independent,
\begin{align}
v_{(1)} &= v_{(1)}^1 \left(  e_1 + i\,  e_2\right)~,\\
v_{(2)} &= v_{(2)}^1 \left(  e_1 + i\,  e_2\right) +  v_{(2)}^2 \left(  e_3 + i\,  e_4\right)~,\\
&~~\vdots\notag
\end{align}
and the diagonal coefficients are positive by linear independence. If we use $\mathrm{O}(m)^{\CC}$ instead, the diagonal entries can be scaled to one. If we are not 
interested in fixing the vectors completely, we may expand 
\begin{equation}\label{eq:complexnulls}
v_{(i)}=v_{(i)}^j\left(e_{2j-1} + i\, e_{2j}\right)
\end{equation}
with the Einstein summation over $j=1,\ldots, \mu$ and use a non-degenerate $\mu\times \mu$ matrix $v_{(i)}^j$. 

There is a manifest $\mathrm{U}(1)^{\mu} \SO(\mu)\subset \SO(m)$ symmetry 
acting on the expansion in terms of $v_{(i)}^j$ in \eqref{eq:complexnulls}. The $\mathrm{U}(1)$ factors are 
complex phase rotations
\begin{equation}
e_{2i-1} + i\, e_{2i} \mapsto e^{i\phi} \left( e_{2i-1} + i\, e_{2i}\right)~,
\end{equation}
and the $\SO(\mu)$ rotates the $e_{2i-1}+i e_{2i}$ in the fundamental representation. The group product $\mathrm{U}(1)^\mu \mathrm{SO}(\mu)$ is not a direct product, 
it is the group generated by the groups $\mathrm{U}(1)^{\mu}$ and $\SO(\mu)$ as subgroups of $\SO(m)$: the set of all possible multiplications between the group elements of the 
subgroups.
As these two subgroups do not commute the multiplication generates $\mathrm{U}(\mu)$, see lemma \ref{lem:uso} in appendix \ref{app:glnc}. 

Similarly, one may fix under $\mathrm{O}(m)^{\CC}$ and the matrix $v_{(i)}^j$ can be made equal to the identity matrix, see appendix \ref{app:glnc}. Let us now turn to 
an ordered set of $\nu \leq 4$ linearly independent complex vectors $\{r_{(i)}\}_{i=1}^{\nu}$ that are mutually orthogonal among themselves and with the previous ordered 
set $\{v_{(i)}\}_{i=1}^{\mu}$ of complex null vectors, but such that the norm of each $r_{(i)}$ is equal to one. Since they are orthogonal to 
the  $\{v_{(i)}\}_{i=1}^{\mu}$, by using $\mathrm{O}(m)$ and the expansion in~\eqref{eq:complexnulls}, we may expand the $r_{(i)}$ as
\begin{equation}\label{eq:rwithnull}
r_{(i)} =\sum_{j=1}^{\mu} B_{(i)}^j \left( e_{2j-1} + i\, e_{2j} \right) + R_{(i)}~,
\end{equation}
where the $R_{(i)}$ do not contain components in the complex span of $\vspan{e_1,\ldots,e_{2\mu}}$. We may use the remaining symmetry $\mathrm{O}(m-2\mu)$ to fix 
the $R_{(i)}$. 

The first $R_{(1)}$ may be brought to the form
\begin{equation}
R_{(1)} = \cosh \zeta_1 \, e_{2\mu+1} + i\, \sinh \zeta_1 \, e_{2\mu+2}~,
\end{equation}
and 
we may choose $\zeta_1$ to be real.
Continuing this way, in a QR decomposition, we may partially fix the $R_{(i)}$ to be expanded in a basis
\begin{equation}
R_{(i)} = \Sigma_{(i)}{}^j e_{2\mu+j}~,
\end{equation}
with an Einstein summation over $j$ and where the matrix $\Sigma_{(i)}{}^j$ is given by the upper-left $\nu \times 2\nu$ submatrix of 
the $4 \times 8$ matrix ($\nu \leq 4$)
\begin{equation}\label{eq:sigmasum}
\begin{aligned}\Sigma_{\text{sup}}  \equiv & 
\left(
\begin{matrix}
\cosh \zeta_1 & i \sinh \zeta_1 & 0 & 0 \\
\sinh \eta_1 \sinh \zeta_1& i \sinh \eta_1 \cosh \zeta_1&
\cosh \eta_1 \cosh \zeta_2 & i \cosh\eta_1 \sinh \zeta_2 \\
0 & 0 & \sinh \eta_2 \sinh \zeta_2& i \sinh \eta_2 \cosh \zeta_2 \\
0&0&0&0 
\end{matrix}\right| \\
& \left| 
\begin{matrix}
 0 & 0 & 0 & 0 \\
 0 & 0 & 0 & 0 \\
\cosh \eta_2 \cosh \zeta_3 & i \cosh\eta_2 \sinh \zeta_3 & 0 & 0 \\
\sinh \eta_3 \sinh \zeta_3& i \sinh \eta_3 \cosh \zeta_3&
\cosh \eta_3 \cosh \zeta_4 & i \cosh\eta_3 \sinh \zeta_4 
\end{matrix} \right)~.
\end{aligned}
\end{equation}
The $\eta_i$ might be fixed to be real or imaginary\footnote{We may choose all coefficients to be real, but not whether $\cosh^2\eta_i$ is larger, equal, or smaller than 
unity.} and the $\zeta_i$ are all real. It might seem that $\Sigma$ is completely fixed and there is no remaining symmetry, but this is not true if $\Sigma$ is 
degenerate. This happens when some of the parameters in $\Sigma$ are zero. The matrix $\Sigma$ has the orthonormal property $\Sigma\Sigma^T=I_{\nu\times\nu}$.

We have now described in general how to fix two ordered sets of vectors 
$\{v_{(i)}\}_{i=1}^{\mu}$ and $\{r_{(i)}\}_{i=1}^{\nu}$ that are orthogonal 
among themselves and 
each other, where the first are null and the latter unit norm, under the 
action of $\mathrm{O}(m)$. Under $\SO(m)$ there might be a sign ambiguity in one 
of 
the components when $2\mu+2\nu=m$. Indeed, for $2\mu+2\nu<m$ one may use a $\SO(m)$ rotation that contains a parity transformation  
perpendicular to the basis, so the sign in the basis is restored. If $2\mu+2\nu=m$ and $\nu\neq0$, we may allow $\eta_i$ to be negative in \eqref{eq:sigmasum}. If 
$\nu=0$ and $2\mu=m$ then we may need to replace $e_{2i-1}+i\, e_{2i}$ with $e_{2i-1}-i\, e_{2i}$ for some $i$ in \eqref{eq:complexnulls}. This sign ambiguity will not 
be present in what follows due to the chirality of Killing spinors.

\subsection{Elements in $N(\mu,\nu)$}
We recall \eqref{eq:Pzatomic} that an element $\Pz \in N(\mu,\nu)$ is $K^{\CC}$-conjugate to
\begin{multline*} 
 \Pz \conjugate{K^{\CC}} \underbrace{\left( e_1 + i \, e_2 \right) \otimes \left( \hat{e}_1 + i\, \hat{e}_2 \right) + \cdots  + \left( e_{2\mu-1} + i \, e_{2\mu}\right) \otimes \left( \hat{e}_{2\mu-1} + i\, \hat{e}_{2\mu} \right)}_{\mu\text{ terms}} \\ +
 \underbrace{\left( e_{2\mu+1} + i \, e_{2\mu+2} \right) \otimes  \hat{e}_{2\mu+1} + \cdots + \left( e_{2\mu+2\nu-1} + i \, e_{2\mu+2\nu}\right) \otimes  
 \hat{e}_{2\mu+\nu}}_{\nu\text{ terms}} ~.
\end{multline*}
The most general $K^{\CC}$ transformation is such that $\Pz$ should be expanded in terms of independent orthogonal complex null 
vectors $\{u_{(i)}\}_{i=1}^{\mu}$ and $\{v_{(i)}\}_{i=1}^{\nu}$ of $\CC^8$ and independent complex null vectors $\{w_{(i)}\}_{i=1}^{\mu}$ and independent 
complex unit-norm vectors $\{r_{(i)}\}_{i=1}^{\nu}$ of $\CC^n$, where the  $w_{(i)}$ and  $r_{(i)}$ are also mutually orthogonal together:
\begin{equation}
\Pz = \sum_{i=1}^{\mu} u_{(i)} \otimes w_{(i)} + \sum_{i=1}^{\nu} v_{(i)} \otimes r_{(i)}~. \label{eq:PIsu}
\end{equation}
This follows by the form given in \eqref{eq:Pzatomic}. Indeed, the action of 
$\SO(8)^{\CC}\times \SO(n)^{\CC}$ preserves the inner product among the vectors 
appearing in \eqref{eq:Pzatomic} or the corresponding ones appearing in 
\eqref{eq:PIsu}. That is, in \eqref{eq:PIsu} we necessarily have
\begin{align}
&u_{(i)}\cdot u_{(j)} = v_{(i)}\cdot v_{(j)} = u_{(i)}\cdot v_{(j)} = 0~, \\
&w_{(i)} \cdot w_{(j)} = w_{(i)}\cdot r_{(j)} = 0~,\\
& r_{(i)}\cdot r_{(j)}=\delta_{ij}~.
\end{align}
Finally, the vectors in \eqref{eq:PIsu} should be linearly independent. 

We define an orthonormal basis 
\begin{equation}
\{\ef_i\}_{i=1}^{2\mu} \oplus \{\es_{i'}\}_{i'=1}^{2\nu} \, ,
\end{equation}
of an orthogonal subspace $\RR^{2\mu}\oplus \RR^{2\nu} \subseteq \RR^8$ and an orthonormal basis 
\begin{equation}
\{\qf_i\}_{i=1}^{2\mu}\oplus\{\qs_{r'}\}_{r'=1}^{2\nu} \, ,
\end{equation} 
of an orthogonal subspace $\RR^{2\mu}\oplus \RR^{2\nu}\subseteq \RR^n$. We will use a basis of null vectors in $\CC^{2\mu+2\nu}\subseteq \CC^8$
\begin{align}
&\left\{ \cb{\ef}{i} \right\}_{i=1}^{\mu}
\oplus 
\left\{\cb{\es}{i'} \right\}_{i'=1}^{\nu} \, ,
 \label{eq:be}
\\\intertext{and a basis of null and orthonormal vectors in $\CC^{2\mu+2\nu}\subseteq \CC^n$}
& \left\{ \cb{\qf}{j} \right\}_{j=1}^{\mu}
\oplus
\left\{  \qs_{r'}  \right\}_{r'=1}^{\nu}
\label{eq:bq}~.
\end{align}
According to the discussion in subsection \ref{sec:complexvectors}, the vectors appearing in the element in \eqref{eq:PIsu} can be fixed 
under $\mathrm{O}(8)\times\mathrm{O}(n)$ (for $m=8$ and $m=n$ in subsection \ref{sec:complexvectors}) so that they are expanded in this basis. 
That is, under $\mathrm{O}(8)\times\mathrm{O}(n)$ the element $\Pz$ can be expanded into
\begin{equation}\label{eq:PzABMN}
\begin{aligned}
P^{Ir} e_I \otimes \hat{e}_r &=
N^{ij} \left( \cb{\ef}{i} \right) \otimes \left( \cb{\qf}{j} \right)
\\ &
  + M^{i'r'} \left( \cb{\es}{i'} \right) \otimes \qs_{r'} 
\\ &
+
A^{ir'} \left( \cb{\ef}{i} \right) \otimes \qs_{r'} 
\\ &
+ 
B^{i'j} \left( \cb{\es}{i'} \right) \otimes \left( \cb{\qf}{j} \right)~.
\end{aligned}
\end{equation}
There are two invariants of the element as written in  \eqref{eq:PzABMN} that identify it as belonging to $N(\mu,\nu)$:
\begin{itemize}
\item The rank $\mu+\nu$ of $\Pz^{Ir}e_I\otimes \hat{e}_r$, and
\item The rank $\nu$ of $\Pz^{Ir}\Pz^{Jr}e_I \otimes e_J$.
\end{itemize}
Note in particular that $\Pz^{Ir}\Pz^{Jr}$ has the same rank as the 
square of the right-hand side of \eqref{eq:Pzatomic}. 

The form of $\Pz$ in \eqref{eq:PzABMN} is the most general element in $N(\mu,\nu)$ up to partial fixing under $K=\SO(8)\times\SO(n)\times\ZZ_2$ for the following reason:
Recall that most of the discussion in subsection \ref{sec:complexvectors} was by using $\mathrm{O}(m)$, here we have so far used $\mathrm{O}(8)\times\mathrm{O}(n)$. If we 
were to use $K$ it might seem that \eqref{eq:PzABMN} still holds up to sign ambiguities in the bases. The mixed parity rotation in $\ZZ_2$ makes this relevant only for the 
null basis \eqref{eq:be} in $\CC^8$. That is, we might need to replace $\cb\ef{i}$ or $\cb\es{i}$ with its conjugate for at most one $i$. If $\mu+\nu<4$ then this is not 
necessary, as one may find an even parity transformation, with one inversion in some complement to the basis \eqref{eq:be} we use, which renders the basis \eqref{eq:be} 
still valid for expanding $\Pz$. Finally, if $\mu+\nu=4$ then the chirality of spinors $\Gamma^{12345678}_{AB}\epsilon_z^{B}=\epsilon^A_z$ guarantees that supersymmetric 
elements in this class are also necessarily of the form \eqref{eq:PzABMN}.

However, we still have a lot of freedom in fixing the element under $K$. We are allowed to use $\mathrm{U}(\mu+\nu)\subset \SO(8)$ on the basis \eqref{eq:be}, and 
$\mathrm{U}(\mu) \times \mathrm{SO}(2\nu) \subset \SO(n)$ on \eqref{eq:bq}. These groups act on the form of $\Pz$ in \eqref{eq:PzABMN} mixing the various coefficients 
but not changing the basis. We will now proceed to fix $\Pz$ in the basis of \eqref{eq:be} and \eqref{eq:bq} by using these groups.

\subsection{Matrix factorizations}

We will use both Takagi's factorization and a singular value decomposition on certain 
coefficients of $\Pz$. Takagi's factorization allows the 
diagonalization  of a symmetric matrix $M M^T$ into a diagonal matrix $D$ via the action of a unitary matrix $S$ by using $D=S M M^T S^T$~\cite{matrix_analysis}. Note that the 
transpose of $S$ is taken instead of the Hermitian transpose. The diagonalization is thus different 
than the spectral decomposition or 
diagonalization by a unitary matrix of a diagonalizable matrix. Takagi's factorization is always possible for symmetric matrices. Furthermore, the diagonal 
elements of $D$ are real, non-negative. On the other hand, the singular value decomposition is the diagonalization of a not necessarily square matrix $N$ under the action of two unitary 
matrices $S_1$ and $S_2$ by using $N \mapsto S_1 N S_2^{\dagger}$, and it is always possible. The diagonal elements are again real and non-negative. 

Consider the square of $\Pz$ as a symmetric complex $(\mu+\nu)\times (\mu+\nu)$ matrix in the basis of  $\{\cb\ef{i}\}_{i=1}^{\mu}$ and $\{\cb\es{i'}\}_{i'=1}^{\nu}$ 
\begin{equation}
\begin{aligned}
\Pz^{Ir} \Pz^{Jr} e_I \otimes e_J &=
\left( M M^T \right)^{i'j'} \left( \cb{\es}{i'} \right) \otimes \left( \cb{\es}{j'} \right) \\
&+ \left( AA^T \right)^{ij}  \left( \cb{\ef}{i} \right) \otimes \left( \cb{\ef}{j} \right) \\
&+ \left( AM^T \right)^{ij'} \left( \cb{\ef}{i} \right) \otimes \left( \cb{\es}{j'} \right) \\ &
+ \left( MA^T \right)^{i'j} \left( \cb{\es}{i'} \right) \otimes \left( \cb{\ef}{j} \right) ~.
\end{aligned}
\end{equation}
We use Takagi's decomposition by using the action of $\mathrm{SU}(\mu+\nu)$ so that
\begin{align}
MM^T &= D~, \text{ (diagonal, real and positive) } \label{eq:takagione}\\
AA^T &= 0~, \label{eq:takagithree}\\
AM^T &= 0\label{eq:takagitwo}~.
\end{align}
We may assert that $D$ does not have zero components because the rank of $\Pz^{Ir}\Pz^{Jr}$ should be preserved under $K^{\CC}$-conjugation\footnote{More precisely, 
Takagi's factorization determines here the split of the basis into $\{\cb\ef{i}\}_{i=1}^{\mu}$ and $\{\cb\es{i'}\}_{i'=1}^{\nu}$, but we have already assumed that 
the split is full rank on the first set.} and is equal to the invariant $\nu$. After this arrangement, the diagonal form of $\Pz^{Ir}\Pz^{Jr}$ is preserved by at least 
$\mathrm{U}(\mu)_{L}\subset \mathrm{U}(\mu+\nu)\subset \SO(8)$ that acts on the $\cb\ef{i}$.  The group that preserves $\Pz^{Ir}\Pz^{Jr}$ might in fact contain an extra 
unitary group if the diagonal elements in $D$ are not all different, but it is not necessary to take this into consideration. After performing Takagi's factorization, the 
full remaining symmetry is at least 
\begin{equation}
\mathrm{U}(\mu)_L \times \mathrm{U}(\mu)_R \times \SO(2\nu) \subset \SO(8)\times \SO(n)~.
\end{equation}
We have labeled the unitary subgroups with L (left) and R (right) to distinguish how they act on $\Pz$, whereas $\SO(2\nu)\subset \SO(n)$ has not been adorned.

The condition $MM^T = D$ can be solved by partially fixing $\mathrm{SO}(2\nu)$. We write
\begin{equation}
M = \sqrt{D} \Sigma~,
\end{equation}
where $\Sigma$ is a $\nu \times 2 \nu$ matrix which satisfies
\begin{equation}
\Sigma \Sigma^T = I_{\nu \times \nu}~,
\end{equation} and on which $\mathrm{U}(\nu)_{L}$ acts on the left in the dual representation and $\mathrm{SO}(2\nu)$ acts on the right. However, we need to 
mod out by the action of the symmetry of the theory, which is precisely the orthogonal group $\SO(2\nu)$ acting on the right of $\Sigma$. By 
using $\mathrm{SO}(2\nu)$ and a 
Gram-Schmidt orthogonalization we can fix $\Sigma$ so that it is the upper-left block of the $4\times 8$ matrix
\begin{equation}
\begin{aligned}\Sigma_{\text{sup}}  = & 
\left(
\begin{matrix}
\cosh \zeta_1 & i \sinh \zeta_1 & 0 & 0 \\
\sinh \eta_1 \sinh \zeta_1& i \sinh \eta_1 \cosh \zeta_1&
\cosh \eta_1 \cosh \zeta_2 & i \cosh\eta_1 \sinh \zeta_2 \\
0 & 0 & \sinh \eta_2 \sinh \zeta_2& i \sinh \eta_2 \cosh \zeta_2 \\
0&0&0&0 
\end{matrix}\right| \\
& \left| 
\begin{matrix}
 0 & 0 & 0 & 0 \\
 0 & 0 & 0 & 0 \\
\cosh \eta_2 \cosh \zeta_3 & i \cosh\eta_2 \sinh \zeta_3 & 0 & 0 \\
\sinh \eta_3 \sinh \zeta_3& i \sinh \eta_3 \cosh \zeta_3&
\cosh \eta_3 \cosh \zeta_4 & i \cosh\eta_3 \sinh \zeta_4 
\end{matrix} \right)~.
\end{aligned}\label{eq:sigmasup}
\end{equation}
This is the same decomposition we described in subsection \ref{sec:complexvectors}. If $\Sigma$ is degenerate, for 
instance if some of the parameters are zero, there is remaining freedom in $\SO(2\mu)$ to further fix its form. This will turn out to be the case 
when we consider in section \ref{sec:solutions} the scalar coset integrability relation. We will then be able to fix $\Sigma$ completely.

We still have a $\mathrm{U}(\mu)_L$ freedom acting on the basis $\cb{\ef}{i}$ and a $\mathrm{U}(\mu)_R$ acting on the basis $\cb\qf{j}$. Their 
action does not spoil the form of $M=\sqrt{D}\Sigma$ with $\Sigma$ described by \eqref{eq:sigmasup}, since we may always use a complementary $\SO(2\nu)$ transformation. 
We use the singular value decomposition on $N$, $N\mapsto S_1 N S_2^{\dagger}$ with $(S_1,S_2)\in \mathrm{U}(\mu)_L\times \mathrm{U}(\mu)_R$, in order to make $N$ diagonal, 
real, non-negative.
We split the basis
\begin{align}
\{ \cb\ef{i} \}_{i=1}^{\mu} & \longrightarrow \{ \cb\efa{i} \}_{i=1}^{\mu_{a}} \oplus \{ \cb\efb{i} \}_{i=1}^{\mu_b} \, ,\\
\{ \cb\qf{i} \}_{i=1}^{\mu} & \longrightarrow \{ \cb\qfa{i} \}_{i=1}^{\mu_{a}} \oplus \{ \cb\qfb{i} \}_{i=1}^{\mu_b} \, ,
\end{align}
so that $N$ is non-zero on the first $\mu_a$ components and zero on the rest of the $\mu_b$ components. There is some remaining symmetry in $K$, an 
anti-diagonal $\mathrm{U}(1)^{\mu_a}$ generated by
\begin{equation}
\left(\cb{\efa}{i}\right) -\left( \cb\qfa{i} \right) ~, i=1,\ldots, \mu_a~
\end{equation}
and a $\mathrm{U}(\mu_b)$, both of which act on the matrices $A$ and $B$. We will not fix $A$ and $B$ though, because the equations of motion will eventually 
force $A=B=0$ and $\mu_b=0$.

We have (partially) fixed the most general element $\Pz\in N(\mu,\nu)$ under the 
action of $K$, which can be summarized as follows: The 
class $N(\mu,\nu)$ of an element $\Pz$ is 
characterised by the rank $\mu+\nu$ of $\Pz$ 
and the rank $\nu$ of $\Pz^{Ir}\Pz^{Jr}$, in which case the element is expanded as in \eqref{eq:PzABMN} in an adapted basis. The 
coefficients $M$ and $A$ in its expansion should satisfy the Takagi relations \eqref{eq:takagione}-\eqref{eq:takagitwo} and $N$ should be diagonal, real, non-negative. 
At this point, we cannot prove that $N$ is strictly positive, as it will turn 
out to be. There is some remaining symmetry acting on $A$ and $B$ and 
possibly on $M$ 
from the right that we do not take advantage of. The basis we are using is
\begin{align}
&\left\{ \cb{\efa}{i} \right\}_{i=1}^{\mu_a}
\oplus \left\{ \cb{\efb}{i} \right\}_{i=1}^{\mu_b}
\oplus 
\left\{\cb{\es}{i'} \right\}_{i'=1}^{\nu}
\\\intertext{in $\CC^8$ and}
&  \left\{ \cb{\qfa}{j} \right\}_{j=1}^{\mu_a}
\oplus \left\{ \cb{\qfb}{j} \right\}_{j=1}^{\mu_b}
\oplus
\left\{  \qs_{r'}  \right\}_{r'=1}^{\nu}
\end{align}
in $\CC^n$, but we will eventually show that $\mu_b=0$ (so $\mu=\mu_a$) and drop the label $a$ on which $N$ is diagonal, real and strictly positive.

\section{Solutions}
\label{sec:solutions}
In this section we impose the equations of motion on the scalar current $\Pz$ whose form is now fixed 
in \eqref{eq:PzABMN}. We first show that the scalar connection 
$Q_z$ is also restricted in form because it has to act on $\Pz$ and preserve the basis that we use for the latter. We may then 
turn to the coset integrability equations in order to show that the form of $\Pz$ is further restricted, for instance it turns out that the matrices $A$ and $B$ 
must be zero. The equations of motion for $\Pz$ and $Q_z$ reduce more and we 
finally arrive at our main result: All solutions decompose into solutions of 
Liouville and $\SU(3)$ Toda systems.

\subsection{Restricting the connection}
\label{sec:restrictingQ}
In order to restrict the possible values of $Q_z$, we make a general analysis of
the equation of motion of $\Pz$ \eqref{eq:DPz}, which we rewrite using the 
notation '$\circ$':
\begin{equation}\label{eq:onlycent}
\partial_{\bar{z}}\Pz + Q_{\bar{z}}\circ \Pz = 0~.
\end{equation} 
We will 
assume that the stabilizer of $\Pz$,
\begin{equation}
\text{stab}(\Pz) = \{ X\in K : X \circ \Pz =0 \}~,
\end{equation}
is trivial. Hence we focus on those elements that act effectively 
on $\Pz$ and enter \eqref{eq:onlycent}.

From \eqref{eq:onlycent} we calculate the equation of motion for $D$
\begin{equation}\label{eq:dD}
\partial_{\bar{z}}\left( \Pz^{Ir} \Pz^{Jr} \right) + \left(  Q_{\bar{z}}^{IK}\delta^{JL} +  Q_{\bar{z}}^{JK}\delta^{IL}\right) \Pz^{Kr}\Pz^{Lr} = 0~.
\end{equation}
Due to the Takagi decomposition, \eqref{eq:dD} involves only the diagonal, positive, real $D$ and we may restrict
\begin{equation}
 \left. Q_{\bar{z}} \right|_{\mathrm{SO}(8)}  \in \mathfrak{u}(1)^{\nu} \oplus \mathfrak{u}(\mu)~. \label{eq:resSO}
\end{equation}
The $\mathfrak{u}(1)^{\nu}$ act on the
\begin{equation}
 \cb{\es}{i'} \mapsto i\left(  \cb{\es}{i'} \right)~
\end{equation}
and enter \eqref{eq:dD} in the form $\partial_{\bar{z}} D + \left. Q_{\bar{z}} \right|_{\mathfrak{u}(1)^{\nu}}\circ D = 0$, while the $\mathfrak{u}(\mu)$ acts on 
the $\cb\ef{i}$ and do not enter \eqref{eq:dD}.

We turn to \eqref{eq:onlycent} again because, since we have restricted $\left. Q_{\bar{z}} \right|_{\mathrm{SO}(8)}$ to a unitary group as in \eqref{eq:resSO}, we may 
assert that
\begin{align}
 \left. Q_{\bar{z}} \right|_{\mathrm{SO}(8)} & \in \mathfrak{u}(1)^{\nu} \oplus \mathrm{u}(1)^{\mu_a} \oplus \mathfrak{u}(\mu_b)~, \label{eq:QIJRestricted}\\
 \left.  Q_{\bar{z}} \right|_{\mathrm{SO}(n)} & \in \mathfrak{u}(1)^{\mu_a} \oplus \mathfrak{u}(\mu_b) \oplus \so(2\nu)~.\label{eq:QrsRestricted}
\end{align}
The two factors of $\mathfrak{u}(1)^{\mu_a}$ act on the positive components of the diagonal, real $N$
\begin{align}
\cb{\efa}{i} &\mapsto i \left( \cb{\efa}{i} \right) \, ,\\
\cb{\qfa}{i} &\mapsto i \left( \cb{\qfa}{i} \right) \, ,
\end{align}
and the remaining $\mathfrak{u}(\mu_b)$ preserves the diagonal form of $N$ in $\Pz$ but acts on the $A$ and $B$. Finally, there is a $\so(2\nu)$ that acts on the 
$\qs_{r'}$ and  thus on $M$ and $A$ from the right. These are the most general subgroups that acting on $\Pz$ should preserve the form of $\partial_{\bar{z}}\Pz$ and 
should thus enter \eqref{eq:onlycent}.

In summary, we have restricted the connection $Q_z$ to take values in
\begin{equation}
Q_z \in  \mathfrak{u}(1)^{\nu}\oplus \mathfrak{u}(1)^{\mu_a}
 \oplus \mathfrak{u}(\mu_b)_L 
 \oplus \mathfrak{u}(\mu_b)_R
 \oplus \so(2\nu)~.
\end{equation}
Explicitly, we have the following expansion
\begin{equation}
\begin{aligned}
Q_z &= q^{(2)}_z  \uw{\es}{i'} + q^{(1a)i}_z \frac12 \left(\uw{\efa}{i} + 
\uw{\qfa}{i}\right) \\
& + q^{(1b)ij}_z \left( \cb\efb{i} \right)\otimes\left( \cbm\efb{j} \right)  \\
&
+ \hat{q}^{(1b)ij}_z \left( \cb\qfb{i} \right)\otimes\left( \cbm\qfb{j} \right) 
\\&+ \Lambda^{rs}_z \qs_{r} \wedge \qs_{s}~.
\end{aligned}
\end{equation}
This may look intimidating at first, but we will soon show that $\mu_b = 0$ and the middle two lines are absent. The components in $\so(2\nu)$ will also be restricted.


\subsection{Integrability of the connection}
Now we are ready to analyze the integrability equations for $Q$ which will restrict $P_z$ even further.
Calculating the right-hand side of the 
integrability equation \eqref{eq:DQIz} for $Q^{IJ}$ as (recall that $N$ is diagonal, real and non-negative)
\begin{equation}
\begin{aligned}
-\Im\left( \Pbz^{Ir} \Pz^{Jr} \right) e_I \otimes e_J&=
-
\Im\big[ 
\left( A^{*} A^T \right)^{ij} \left( \cbm\ef{i} \right) \otimes \left( \cb\ef{j} \right) \\
&+ \left( A^{*} M^T \right)^{ij'} \left( \cbm\ef{i} \right) \otimes \left( \cb\es{j'} \right) \\
&+ \left( A^{*} M^T \right)^{ij'}  \left( \cbm\es{j'} \right)  \otimes \left( \cb\ef{i} \right) \\
&+ \left( M^{*} M^T \right)^{i'j'}  \left( \cbm\es{i'} \right)  \otimes \left(\cb\es{j'} \right) \\
&+ 2 \left( NN^T \right)^{i}  \left( \cbm\ef{i} \right) \otimes \left( \cb\ef{i} \right) \\
&+ 2 \left( N B^T \right)^{ij'}  \left( \cbm\ef{i} \right) \otimes  \left( \cb\es{j'} \right) \\
&+ 2 \left( N B^T \right)^{ij'}   \left( \cbm\es{j'} \right) \otimes   \left( \cb\ef{i} \right) \\
&+ 2 \left( B^{*} B^T \right)^{i'j'}  \left( \cbm\es{i'} \right) \otimes  \left( \cb\es{j'} \right) \big]~. 
\end{aligned}
\label{eq:rhsint1}
\end{equation}
From the form of $\left. Q \right|_{\SO(8)}$ in \eqref{eq:QIJRestricted} we deduce that
\begin{align}
M^{*}M^T + 2 B^{*} B^T &\in \mathfrak{u}(1)^{\nu}~, \label{eq:res2}\\
A^{*}M^T + 2 N B^T &= 0~.\label{eq:res3}
\end{align}
Similarly, we calculate the right-hand side  of the integrability equation \eqref{eq:DQrz} for $Q^{rs}$ as
\begin{equation}
\begin{aligned}
-\Im \left( \Pbz^{Ir} \Pz^{Is} \right) \hat{e}_r \otimes \hat{e}_s &=
-
2\Im\big[ 
 \left( A^\dagger A \right)^{rs} \qs_{r} \otimes \qs_{s} \\
& +  \left( A^{\dagger} N \right)^{rj} \qs_{r} \otimes \left(
\cb\qf{j} \right) \\
&+  \left( N^{T} A \right)^{jr} \left( \cb\qfb{j} \right) \otimes \qs_{r} \\
&+  i\, \left(N^2 \right)^{ii}  \uw\qf{i}\\
&+  \left( M^{\dagger} M \right)^{rs} \qs_{r} \otimes \qs_{s}\\
&+  \left( M^{\dagger} B \right)^{rj} \qs_{r} \otimes \left( 
\cb\qf{j} \right)\\
&+  \left( B^{\dagger} M \right)^{jr} \left( \cbm\qf{j} \right) \otimes \left( \cb\qf{r}  \right)  \\
&+  \left( B^{\dagger} B \right)^{ij} \left( \cbm\qf{i} \right) \otimes \left( \cb\qf{j} \right) \big] \, .
\end{aligned}\label{eq:rhsint2}
\end{equation}
From the form of $\left. Q \right|_{\SO(n)}$ in \eqref{eq:QrsRestricted} we deduce that
\begin{align}
A^{\dagger} N + M^{\dagger} B &= 0~. \label{eq:AdN}
\end{align}

We first show that $B=0$ and that $A$ is further restricted. Multiplying \eqref{eq:AdN} with $M^{*}$ from the left gives
\begin{equation}
\left( M A^T \right)^{*} N + \left( M M^T \right)^{*} B =0~.
\end{equation}
However, the Takagi relations (see \eqref{eq:takagione} and 
\eqref{eq:takagitwo}) are $M A^T = 0$ and that $MM^T = D$ is invertible. 
Hence $B=0$. We also have that $N$ is invertible only in the first $\mu_a$ 
diagonal components. With 
$B=0$, \eqref{eq:AdN} becomes $A^{\dagger}N=0$, hence $A^{ir'} = 0$ for $i= 1,\ldots , \mu_a$. 
Using $B=0$ in \eqref{eq:res3}, one finds that
\begin{equation}
A M^{\dagger} =0 \label{eq:AMd}
\end{equation}
and recall the Takagi condition \eqref{eq:takagitwo} on $A$:
\begin{equation}
A M^T =0 \label{eq:AMt}~.
\end{equation}
We now turn to solving $M$, which will later lead us to $A=0$.

By using $B=0$, \eqref{eq:res2} states that $M^{*} M^T$ is diagonal, which after the Takagi relation $M=\sqrt{D}\Sigma$ becomes
\begin{equation}
\Sigma^{*} \Sigma^T = \text{diagonal}~. \label{eq:sigmastar}
\end{equation}
Recall that we have partly fixed $\Sigma$ in \eqref{eq:sigmasup} by (partly) using $\mathrm{SO}(2\mu)$. The condition \eqref{eq:sigmastar} is satisfied provided that 
the parameters in $\Sigma_{\text{sup}}$ \eqref{eq:sigmasup} satisfy
\begin{align}
 \sinh \eta_1 \sinh\zeta_1 &=0~,\\
 \sinh \eta_2 \sinh\zeta_2 &=0~,\\
 \sinh \eta_3 \sinh\zeta_3 &=0~.
\end{align}
When these hold, $\Sigma$ becomes degenerate and can be reduced to a non-degenerate block form by use of $\SO(2\mu)$. In particular, we can reduce $\Sigma$  to 
be of the form of a $(\nu_r + \nu_c)\times \left( \nu_r + 2 \nu_c \right)$ matrix with values
\begin{equation}
\Sigma = \begin{pmatrix} I_{\nu_r \times \nu_r} & 0 &  0 &0  & 0 & \cdots \\
 0 & \cosh \zeta_1 & i \sinh \zeta_1 & 0 & 0  & \cdots \\
 0 & 0 & 0 & \cosh \zeta_2 & i \sinh \zeta_2 & \cdots \\
 \vdots &&&&&\ddots&\end{pmatrix}~.
\end{equation}
We may also write for the matrix $D$
\begin{equation}
D = \text{diag}\left( D_r^1, \ldots , D_r^{\nu_r}, D_c^1, \ldots, D_c^{\nu_c} \right)~,\label{eq:Drc}
\end{equation}
in which case $M$ is now given by
\begin{equation}
M = \begin{pmatrix} \sqrt{D_r} & 0 &  0 &0  & 0 & \cdots \\
 0 & \sqrt{D_c^1} \cosh \zeta_1 & i\, \sqrt{D_c^1} \sinh \zeta_1 & 0 & 0  & \cdots \\
 0 & 0 & 0 & \sqrt{D_c^2}\cosh \zeta_2 & i\, \sqrt{D_c^2} \sinh \zeta_2 & \cdots \\
 \vdots &&&&&\ddots&\end{pmatrix}~.\label{eq:Msol}
\end{equation}
We may now return to imposing both \eqref{eq:AMd} and \eqref{eq:AMt} with this particular $M$ and we arrive at $A=0$.

Let us summarise what we found: By using the general element $\Pz\in N(\mu,\nu)$ given as \eqref{eq:PzABMN}, the form of $Q_z$ in \eqref{eq:QIJRestricted} and \eqref{eq:QrsRestricted}, 
and the integrability equations for the connection $Q$, the right-hand side of which are in \eqref{eq:rhsint1} and \eqref{eq:rhsint2}, the most general element $\Pz$ up 
to the action of $K$ is shown to be equal to
\begin{equation}\label{eq:PzMN}
P^{Ir} e_I \otimes \hat{e}_r =
N^{i} \left( \cb{\ef}{i} \right) \otimes \left( \cb{\qf}{i} \right)
  + M^{i'r'} \left( \cb{\es}{i'} \right) \otimes \qs_{r'} ~,
\end{equation}
where all the $\mu$ components $N^i$ are positive real and $M^{i'r'}$ is as in \eqref{eq:Msol}. Here, we have set $\mu_b=0$ and dropped any label $a$ from the 
basis $\cb{\efa}{i}$ and $\cb\qfa{i}$. Indeed, with $A=B=0$ the diagonal non-negative $N$ should be strictly positive in order for $\Pz$ to have rank $\mu+\nu$. The 
$(\nu_r+\nu_c)\times(\nu_r+2\nu_c)$ matrix $M$ has a special decomposed block diagonal form according to \eqref{eq:Msol}. We thus say that the element $\Pz$ belongs 
to the refined class $N(\mu,\nu_r,\nu_c)$,
\begin{equation}
\Pz \in  N(\mu,\nu_r,\nu_c)~,
\end{equation}
a filtering of the elements we were thus far considering in $N(\mu,\nu)$.

\subsection{Field equations and Toda blocks}
The block form of $M$ in \eqref{eq:Msol} suggests that we refine the basis we are using. We split the basis as
\begin{align}
\left\{ \cb\es{i'} \right\}_{i'=1}^{\nu} &\longrightarrow \left\{ \cb{\er}{i'} \right\}_{i'=1}^{\nu_r} \oplus  \left\{ \cb{\ec}{i'} \right\}_{i'=1}^{\nu_c} \\
\left\{ \qs_{r'} \right\}_{r'=1}^{2\nu} &\longrightarrow \left\{ \qr_{r'} \right\}_{r'=1}^{\nu_r} \oplus  \left\{ \qc_{r'} \right\}_{r'=1}^{2\nu_c} \oplus \left\{ \text{rest} \right\}~, 
\end{align}
where by ``rest'' we mean those orthonormal basis vectors in $\RR^n$ that do not appear in $\Pz$. That is, we are now using the basis vectors
\begin{align}
& \left\{ \cb\ef{i} \right\}_{i=1}^{\mu} \oplus  \left\{ \cb\er{i} \right\}_{i=1}^{\nu_r} \oplus  \left\{ \cb\ec{i} \right\}_{i=1}^{\nu_c} \\
\intertext{in $\CC^8$ and}
& \left\{ \cb\qf{i} \right\}_{i=1}^{\mu} \oplus  \left\{ \qr_{r'} \right\}_{r'=1}^{\nu_r} \oplus  \left\{ \qc_{r'} \right\}_{r'=1}^{2 \nu_c} 
\end{align}
in $\CC^{n}$. The expansion of an element $\Pz\in N(\mu,\nu_r,\nu_c)$ in this basis is
\begin{equation}
\begin{aligned}
\Pz &= 
\sum_{i=1}^{\mu} N^{i} \left( \cb{\ef}{i} \right) \otimes \left( \cb{\qf}{i} \right)
 \\& + \sum_{i=1}^{\nu_r} \sqrt{D_r^i} \left( \cb{\er}{i} \right) \otimes \qr_{i} \\&
  +\sum_{i=1}^{\nu_c} \left( \cb\ec{i}\right)\otimes \left( \sqrt{D_c^i} \cosh \zeta_i \, \qc_{2i-1} + i\, \sqrt{D_c^{i}} \sinh\zeta_i \, \qc_{2i} \right)~.
\end{aligned}
\end{equation}
As a matrix in the $\{e_I \otimes \hat{e}_r\}$ basis of $\CC^8\otimes\CC^n$, the element $\Pz^{Ir} e_I \otimes e_r$ is block diagonal with $\mu$, $\nu_r$ and $\nu_c$ 
blocks of (respectively) the type
\begin{equation}
\begin{pmatrix}
N^i & i\, N^i \\
i\, N^i & - N^i
\end{pmatrix}~, ~~
\begin{pmatrix}
\sqrt{D_r^i} \\ i \sqrt{D^i_r}
\end{pmatrix}~,~~
\begin{pmatrix}
\sqrt{D_c^i} \cosh\zeta_i& i\, \sqrt{D_c^i} \sinh\zeta_i \\
i\, \sqrt{D_c^i} \cosh\zeta_i & - \sqrt{D_c^i} \sinh\zeta_i
\end{pmatrix}~
\label{eq:todablocks}
\end{equation}
and trailing zeros\footnote{Trailing zeros evidently happen when $\mu+\nu_r+\nu_c=4$ or $2\mu+2\nu_c+\nu_r = n$.}. The $N^i$, $D^i_r$, $D^i_c$ and $\zeta_i$ are 
real functions of $z$ and $\bar{z}$. The right-hand side of the integrability equations \eqref{eq:rhsint1} and \eqref{eq:rhsint2}  become
\begin{equation}
\begin{aligned}
-\Im\left( \Pbz^{Ir} \Pz^{Jr} \right) e_I \otimes e_J&=
- D^i_r \uw\er{i} - D^i_c  \cosh 2\zeta_i \,   \uw\ec{i} \\
& - 2 \left( N^i \right)^2  \uw\ef{i} \label{eq:unuc1}
\end{aligned}
\end{equation}
and
\begin{equation}
-\Im \left( \Pbz^{Ir} \Pz^{Is} \right) \hat{e}_r \otimes \hat{e}_s = -2 \left( N^i \right)^2 \uw\qf{i} -D^i_c\sinh 2\zeta_i \, \uw\qc{i} ~, \label{eq:unuc}
\end{equation}
respectively. With this, we can now reduce the connection to be
\begin{equation}
\begin{aligned}
Q_z &= q^{(1)i}_z \frac12 \left( \uw\ef{i} + \uw\qf{i} \right) + q^{(r)i}_z \uw\er{i}\\& + q^{(c)i}_z\, \uw\ec{i} + \hat{q}^{(c)i}_z \, \uw\qc{i}~.
\end{aligned}
\end{equation}
Any other component can be gauged to zero, because the curvature of the 
connection is non-trivial only in these components. We observe that the blocks 
decouple and we can solve the equations of motion separately for each sector $N^i$, $D^{i}_r$ and $(D^i_c,\zeta_i)$. We call each independent sector a Toda block.

For the $N$ functions we have
\begin{align}
\partial_{\bar{z}} N^i + i q^{(1)i}_{\bar{z}} N^i &=0\, ,\\
\Im\left( \partial_{\bar{z}} q^{(1)i}_z \right) &= - 4 (N^i)^2 \, ,
\end{align}
from which we derive the Liouville equation
\begin{equation}
\partial_z \partial_{\bar{z}} \ln N^i = 4 (N^i)^2~.
\label{eq:liouville1}
\end{equation}
We proceed by writing the equations of motion that involve $D$
\begin{align}
\partial_{\bar{z}} D_r^i + 2 i q^{(r)i}_{\bar{z}} D_r^i & = 0 \, ,\\
\partial_{\bar{z}} D_c^i + 2 i q^{(c)i}_{\bar{z}} D_c^i & = 0 \, , \label{eq:dDc}
\end{align}
with integrability conditions
\begin{align}
\Im\left( \partial_{\bar{z}}  q^{(r)i}_z \right) &= - D_r^i~, \\
\Im\left( \partial_{\bar{z}} q^{(c)i}_z \right) &= - D_c^i  \cosh 2\zeta_i~. \label{eq:toda1}
\end{align}
These give the equations
\begin{align}
\partial_z \partial_{\bar{z}} \ln D_r^i &= 2 D_r^i~, \label{eq:liouville2}\\
\partial_z \partial_{\bar{z}} \ln D_c^i &= 2 D_c^i \cosh 2\zeta_i~.\label{eq:toda2}
\end{align}
Finally, the equations of motion for the $\zeta_i$ can be found from $\mathcal{D}_{\bar{z}} \Sigma = 0$ and are
\begin{subequations}
\begin{align}
\partial_{\bar{z}} \cosh\zeta_i - i\, q^{(c)i}_{\bar{z}} \cosh \zeta_i + i\,  \hat{q}^{(c)i}_{\bar{z}} \sinh\zeta_i &= 0~, 
\\
\partial_{\bar{z}} \sinh\zeta_i - i\, q^{(c)i}_{\bar{z}} \sinh \zeta_i + i\,  \hat{q}^{(c)i}_{\bar{z}} \cosh\zeta_i &= 0~, 
\end{align}\label{eq:toda3}
\end{subequations}
while there is a remaining integrability equation,
\begin{equation}
\Im\left( \partial_{\bar{z}} \hat{q}^{(c)i} \right) \uw\qc{i} = -\frac12 \Im \left( \Sigma^{\dagger} D_c \Sigma^T \right)^{rs} \qs_{r} \wedge \qs_{s}~,
\end{equation}which yields
\begin{equation}
\Im\left( \partial_{\bar{z}} \hat{q}^{(c)i}_z  \right) = - D_c^i \sinh 2 \zeta_i~.\label{eq:toda4}
\end{equation}
The equation \eqref{eq:liouville2} for the $D_r$ is another copy for the Liouville equation, whereas \eqref{eq:dDc}, \eqref{eq:toda1}, \eqref{eq:toda3} and 
\eqref{eq:toda4} describe an $\SU(3)$ Toda system. All equations are integrable 
and indeed solvable. Furthermore, having solved the coset integrability 
equations, $\mathcal{V}^{-1} \grad \mathcal{V}  = P+Q$ can be integrated to 
obtain the coset representative $\mathcal{V}$.

\subsection{General solutions}
We are interested in solving the equations of the Toda blocks in a punctured bounded domain of the complex plane. The solution to the positive Liouville 
modes $N^i$ in \eqref{eq:liouville1} is
\begin{equation}
\left( N^i \right)^2 =\frac{1}{4} \frac{ \partial_z f_i \partial_{\bar{z}} \bar{f}_{i}}{\left(1-|f_i|^2\right)^2} = -\frac14 \partial_z\partial_{\bar{z}} \ln\left(1-|f_i(z)|^2\right)~.
\label{eq:liouvillesol1}
\end{equation}
Similarly, the solution to the positive $D_r^i$ is
\begin{equation}
D_r^i =  \frac{ \partial_z g_i \partial_{\bar{z}} \bar{g}_{i}}{\left(1-|g_i|^2\right)^2}= -\partial_z\partial_{\bar{z}}\ln\left(1-|g_i(z)|^2\right)~.
\label{eq:liouvillesol2}
\end{equation}
The complex functions $f_i(z)$ and $g_i(z)$ are allowed here to be meromorphic. However, only simple poles of the $f_i$ and $g_i$  give smooth solutions in 
\eqref{eq:liouville1} and \eqref{eq:liouville2}. A concrete answer on the nature 
of the singularities can be given by requiring finite coset space charge, 
which we do not analyze here.
The solutions we presented above are a rewriting of Liouville's general solution such that the modes are manifestly positive. As such, the domain of the solution 
should not contain roots of $1-|f_i(z)|^2=0$ or $1-|g_i(z)|^2=0$.

In order to solve the $\SU(3)$ Toda system, we should write it canonically. In 
particular, we should diagonalize the first-order equations for $D_c^i$ and 
$\cosh\zeta_i$. 
Define
\begin{align}
\Phi_1^i & \equiv \frac12 \sqrt{D_c^i} e^{\zeta_i}~, \\
\Phi_2^i & \equiv  \frac12 \sqrt{D_c^i} e^{-\zeta_i}~.
\end{align}
Their gauge-invariant equations of motion are derived from \eqref{eq:dDc} and \eqref{eq:toda3}:
\begin{align}
\partial_{\bar{z}} \Phi_1^i + i\, \hat{q}^{(c)i}_{\bar{z}} \Phi_1^i &= 0~,\\
\partial_{\bar{z}} \Phi_2^i - i\, \hat{q}^{(c)i}_{\bar{z}} \Phi_2^i &= 0~,
\end{align}
while \eqref{eq:toda4} becomes
\begin{equation}
\Im\left( \partial_z \hat{q}^{(c)i}_{\bar{z}}  \right) = 4\left(\Phi_1^i\right)^2 - 4\left(\Phi_2^i\right)^2 ~.\label{eq:toda5}
\end{equation}
The connection $q^{(c)i}_z$ can thus be found from \eqref{eq:toda1} once we solve the above three equations and $\hat{q}^{(c)i}_z$ can be found from \eqref{eq:toda5} 
if we have a solution for the $\Phi_1^i$ and $\Phi_{2}^i$. We gauge fix $(\Phi_1^i,\Phi_{2}^i)$ to be real and positive. We can then eliminate $\hat{q}^{(c)i}_{\bar{z}}$ 
from the three equations:
\begin{align}
\partial_z\partial_{\bar{z}} \ln \Phi_1^i &=2 \left(\Phi_1^i\right)^2 - \left(\Phi_2^i\right)^2 ~,\\
\partial_z\partial_{\bar{z}} \ln \Phi_2^i &=2 \left(\Phi_2^i\right)^2 - \left(\Phi_1^i\right)^2 ~.
\end{align}
This has the form of the $\SU(3)$ Toda field equations
\begin{equation}
\partial_z\partial_{\bar{z}} \ln \Phi_a^i =\sum_b C_{ab} \left( \Phi_a^i\right)^2~,
\end{equation}
where $C_{ab}$ is the $\SU(3)$ Cartan matrix.


A simple form for the general solution of the $\SU(N)$ Toda equation in two-dimensional Minkowski spacetime and with negative coupling constant that is reminiscent of the Liouville 
solution was derived in \cite{dunne_self-dual_1995} from Kostant's solution. We amend that solution for $N=3$, Euclidean signature and positive coupling constant:
\begin{align}
\left(\Phi^i_1\right)^2 &=- \frac12 \partial_z \partial_{\bar{z}} \ln \left(\begin{pmatrix} 1  \\ -\bar{F}_{i}(\bar{z}) \\-\bar{G}_{i}(\bar{z})\end{pmatrix}^T
\begin{pmatrix} 1  \\ {F}_{i}({z}) \\{G}_{i}(z)\end{pmatrix}  \right)
\label{eq:todasol1}~,
\\
\left(\Phi^i_2\right)^2 &=- \frac12 \partial_z \partial_{\bar{z}} \ln \det \left(\begin{pmatrix} 1 & 0 \\ -\bar{F}_{i}(\bar{z}) & - \partial_{\bar{z}}\bar{F}_{i}(\bar{z})\\-\bar{G}_{i}(\bar{z})&-\partial_{\bar{z}}\bar{G}_{i}(\bar{z})\end{pmatrix}^T
\begin{pmatrix} 1 & 0 \\ {F}_{i}({z}) &  \partial_{{z}}{F}_{i}(z)\\{G}_{i}(z)&\partial_{{z}}{G}_{i}({z})\end{pmatrix} \right)~.
\label{eq:todasol2}
\end{align}
Note that we keep the index $i$ of the $\nu_c$ copies. For $G_i(z)=0$ the solution indeed matches Liouville's. Similarly to the Liouville solutions, we may 
allow the functions to be 
meromorphic but restrictions should be applied to ensure that the coset charge is finite.

The full connection $Q_z$ can always be solved from the Toda block solutions of this section. What finally remains is the Einstein equation. Recall that 
its non-trivial component is given by \eqref{eq:PPz} and allows us to solve for the conformal factor in 
the metric given by an exponential of $\rho$. Not only can $\rho$ be solved for each Toda block, the Einstein equation is linear in the block decomposition:
\begin{equation}
 -2 \partial_z\partial_{\bar{z}}\rho = 4 \sum_{i=1}^{\mu}\left(N^i\right)^2
+ 2 \sum_{i=1}^{\nu_r} D_r^i  +2 \sum_{i=1}^{\nu_c} D_c^i \cosh 2 \zeta_i~.
\end{equation}
We have presented the Toda block solutions in the form $\partial_z\partial_{\bar{z}}(\cdots)$ for this reason: the Einstein equation is thence easily integrated. By 
using the explicit solutions \eqref{eq:liouvillesol1}, \eqref{eq:liouvillesol2}, \eqref{eq:todasol1} and \eqref{eq:todasol2}, the solution up to boundary terms is given by
\begin{equation}
\begin{aligned}
\rho &= \frac12 \sum_{i=1}^{\mu}\ln \left( 1-|f_i(z)|^2\right) + 
 \sum_{i=1}^{\nu_r}\ln \left( 1-|g_i(z)|^2\right) 
\\
&+ \sum_{i=1}^{\nu_c}\ln \det \left\{ 
\begin{pmatrix} 1  \\ -\bar{F}_{i}(\bar{z}) \\-\bar{G}_{i}(\bar{z})\end{pmatrix}^T
\begin{pmatrix} 1  \\ {F}_{i}({z}) \\{G}_{i}(z)\end{pmatrix} 
\times \right.
\\
& 
\left.
\quad\begin{pmatrix} 1 & 0 \\ -\bar{F}_{i}(\bar{z}) & - \partial_{\bar{z}}\bar{F}_{i}(\bar{z})\\-\bar{G}_{i}(\bar{z})&-\partial_{\bar{z}}\bar{G}_{i}(\bar{z})\end{pmatrix}^T
\begin{pmatrix} 1 & 0 \\ {F}_{i}({z}) &  \partial_{{z}}{F}_{i}(z)\\{G}_{i}(z)&\partial_{{z}}{G}_{i}({z})\end{pmatrix} 
\right\}~.\end{aligned}
\end{equation}
With this, we have locally found the metric \eqref{eq:ultrastatic}
of the most general timelike supersymmetric solution. The scalar curvature can then be computed from $R=2e^{-2\rho}\partial_z\partial_{\bar{z}}\rho$.

If the meromorphic functions are defined at infinity, in which case there are necessarily singularities elsewhere on the Riemann sphere, the function $\rho$ 
will also have a well-defined limit at infinity. As an example let us look at the simplest solution, namely $N(1,0,0)$, for the metric, but a similar analysis applies 
to the $N(0,1,0)$ solution. The metric is 
of the form
\begin{equation}
\grad s^2= \grad t^2 - \left(1-|f(z)|^2\right) \grad z \grad \bar{z}~.
\end{equation}
If $f(z)$ has a simple pole only at the origin of the Riemann sphere, then $f(z)=a+c/z$. If we further choose $c>0$ and $a=0$, then the metric becomes
\begin{equation}
\grad s^2 = \grad t^2 - \left(1-\frac{c^2}{r^2} \right) \left( \grad r^2 + r^2 \grad \theta^2\right)~.
\end{equation}
We may consider then the exterior of $r=c$ and the metric is manifestly asymptotically flat. We leave a more thorough analysis of the properties of the 
solutions for future work.

We note that the half-BPS solutions of $\SO(8,n)$ with $n>2$ are always given by the Toda blocks $N(1,0,0)$, $N(0,1,0)$ and $N(0,0,1)$. Other examples are given in the 
following:
\begin{example}
The only timelike supersymmetric solution of $\SO(8,1)$ supergravity is $N(0,1,0)$ and it preserves 8 real supersymmetries. The timelike supersymmetric solutions 
of $\SO(8,2)$ supergravity are given by $N(1,0,0)$, $N(0,1,0)$, $N(0,2,0)$ and $N(0,0,1)$. They preserve 8, 8, 4 and 8 real superymmetries respectively . The timelike 
supersymmetric solutions of $\SO(8,3)$ supergravity are given by those of $\SO(8,2)$ and the solutions $N(0,3,0)$, $N(1,1,0)$ and $N(0,1,1)$ that preserve respectively 
2, 4 and 4 real supersymmetries. In each case, we need to fit the Toda blocks \eqref{eq:todablocks} in a $8\times n$ matrix $\Pz^{Ir}$.
\end{example}
\begin{example}
The supersymmetric solution presented in subsection 5.2 of \cite{deger_supersymmetric_2010} is restricted to $n\leq4$. Since it has $Q^{rs}_z=0$, we 
identify it initially with the $N(0,\nu_r,0)$ class. Then $P_z^{Ir}$ is taken proportional to a constant matrix $U^{ir}$, $\Pz^{Ir}\sim \mathbb{P}^{Ii} U^{ir}$, 
with $U^{\dagger}U=I_{n\times n}$ where $\mathbb{P}^{Ii}$ is the null basis $\{\cb{e}{i}\}_{i=1}^4$, see also the discussion around \eqref{eq:degerF} in 
appendix \ref{app:timespin}. The matrix $U$ is thus effectively proportional to the $n\times n$ identity matrix and we identify\footnote{The 
reduced equations of motion of \cite{deger_supersymmetric_2010} match with ours, as they should, provided we identify the fields $\zeta$ and $g$ that 
appear there according to $e^{\rho}\zeta=2\sqrt{D_r}$ and $q_z^{(r)i}=g \zeta^2$, but we solve them essentially differently. Note also that in \cite{deger_supersymmetric_2010} the local coset symmetry $\SO(8)$ breaks into $\SO(2)\times 
 \SO(6)$, whereas here it is broken to $\mu+\nu_r+\nu_c\leq 4$ copies of $\SO(2)$.} the solution with $N(0,n,0)$ and with all Toda fields $D^i_r$ equal, that 
 is $D_r^i=D_r$ for $i=1,\ldots,n$. 
\end{example}
\section{Discussion and comments}
In this article we have classified and explicitly obtained all timelike 
supersymmetric solutions of three-dimensional half-maximal ungauged 
supergravity. 
The structure of the supersymmetric solutions that we found, which are in blocks 
of Liouville and $\SU(3)$ Toda systems, is new and surprisingly simple. With 
the null supersymmetric waves having already been solved in 
\cite{deger_supersymmetric_2010}, all supersymmetric solutions of the ungauged 
$\SO(8,n)$ theory are now known.

It may at first seem surprising  
that the supersymmetric solutions of half-maximal D=3 supergravity  
have only been classified and solved for more than 30 years after its construction  
 in \cite{marcus_three-dimensional_1983}. It is therefore of 
importance to trace our method and pinpoint its novelty. The classification 
under $K^\CC$, as introduced first in \cite{de_boer_classifying_2014}, 
characterizes classes uniquely by two invariants: the rank $\mu+\nu$ of $\Pz$ 
and the rank $\nu$ of $\Pz^{Ir}\Pz^{Jr}$. When we refine this classification 
with respect to the real symmetry of the theory, these two invariants are 
preserved. One could do away with the detour into the indecomposable types of 
the complex group and with some work arrive at the same classes 
$N(\mu,\nu_r,\nu_c)$ provided one uses the same two invariants.

Given the elements of $\Pz$ in these classes, and in particular due to the 
invariant $\nu$, we were naturally led to the use of Takagi's factorization. 
This is a rather uncommon method compared to the spectral or eigenvalue 
decomposition that does not preserve the invariant $\nu$. Furthermore, 
an eigenvalue decomposition or singular value decomposition on $\Pz$ would have 
been impossible unless one enlarged the symmetry of the theory, for instance one 
might consider $\SO(2\nu_r)\rightarrow \SU(2\nu_r)$ or $\SO(n)\rightarrow 
\SU(n)$. The subsequent factorization of $N$ that we employed by using the 
singular value decomposition 
comes as a concession, in the sense that we are manifestly allowed to use it 
after Takagi's factorization. We finally enforced the equations of motion, which 
further reduced the possible form of the coset representative. 

The success of our method seems promising in employing it perhaps to the maximally supersymmetric supergravity. The classification under the complex local symmetry 
was already achieved in \cite{de_boer_classifying_2014} and perhaps finding all 
elements up to the real local symmetry is possible. Certainly though, the 
$\SO(8,n)$ 
representations appearing here are easier to work with. Another interesting extension of our work is to examine interesting monodromies, similar to the reasoning in 
\cite{de_boer_classifying_2014}. One now has the advantage that all solutions are known and requiring single-center monodromies is straightforward.

More generally, one would like to have a more thorough analysis of solutions to the Toda blocks and their geometric analysis. We have already noted that if the 
holomorphic functions are well-defined at infinity, then under some conditions one can conformally compactify the space that is now asymptotically flat. The fundamental 
BPS states, that are the non-smooth single-center solutions, are particularly interesting also for quantum considerations. Even classically, the smooth solutions of the 
theory are the smearing of the fundamental solutions, and a careful analytic study of the Toda blocks is lacking in our work.

Let us briefly comment on one more extension of our work. The success of our method might imply that it has a place in the non-abelian gauged version of the theory \cite{nicolai_n8_2001}, in which a subgroup of the global SO(8,n) is gauged by  
Chern-Simons gauge fields. Although the gauged theory upon 
imposing supersymmetry possesses a corresponding structure, the starting equation \eqref{eq:algsusy} is deformed in such a way that the nilpotency argument can no longer be 
applied. It would be interesting to find a solution to this problem.
\section*{Acknowledgments}
G.M. is supported fully and N.S.D. and {\"O}.S. are supported partially by T{\"U}B{\.I}TAK 
project 113F034.

\appendix

\section{Spin structure of timelike backgrounds}
\label{app:spin}
\subsection{Representations}
For the mostly minus metric 
\begin{equation}
\grad s^2=+\grad t^2 - e^{2\rho(x,y)}\left(  \grad x^2 + \grad y^2 \right)~,
\end{equation}
we use the vielbein $\theta^{\flatc{0}}=\grad t$, $\theta^{\flatc{1}}=e^{\rho}\grad x$ and $\theta^{\flatc2}=e^{\rho}\grad y$. From 
\begin{equation}
\grad \theta^{\flatc{a}} + \omega^{\flatc{a}}{}_{\flatc{b}}\wedge \theta^{\flatc{b}} = 0~,
\end{equation}
we find the non-zero spin coefficient
\begin{equation}
\omega_{\flatc{i}\flatc{j}} = - \partial_j \rho \,  \grad x^i + \partial_i \rho \,  \grad x^j \quad \quad i,j=1,2~.
\end{equation}
The Riemann curvature has non-zero component, in flat coordinates,
\begin{equation}
R_{\flatc{1}\flatc{2}\flatc{1}\flatc{2}}= e^{-2\rho}\partial_i\partial_i\rho~,
\end{equation}the non-zero components of the Ricci tensor is
\begin{equation}
R_{\flatc{i}\flatc{j}}=-e^{-2\rho}\partial_k\partial_k \rho \, \delta_{ij}~
\end{equation}
and the Ricci scalar is $R=2 e^{-2\rho} \partial_k\partial_k \rho$.

By using the complex coordinate $z=x+i\, y$ and $\partial_z=\frac12\left(\partial_x-i\partial_y\right)$, we define the complex components $(\phi_z,\phi_{\bar{z}})$ for a 
one-form with $\phi_t=0$,
\begin{equation}
\phi_x \grad x + \phi_y \grad y = \frac12 \left( \phi_x - i \, \phi_y \right) \grad z + \frac12 \left( \phi_x + i \, \phi_y \right) \grad \bar{z} = \phi_z \grad z + \phi_{\bar{z}} \grad \bar{z}~.
\end{equation}
For two such one-forms, we have
\begin{align}
\phi_{\bar{z}} \chi_z = \frac14 \phi_i \chi_i - \frac14 i \epsilon_{ij} \phi_i \chi_j 
\end{align}
with the antisymmetric $\epsilon_{12}=1$. That is, both the inner product and the wedge of the two one-forms $\phi$ and $\chi$ are recovered from the 
Hermitian product of complex functions $\phi_{\bar{z}} \chi_z$.

We use the three-dimensional gamma matrices
\begin{equation}
\gamma^0 = \begin{pmatrix} 0 & 1 \\ -1 & 0 \end{pmatrix}~,
\gamma^1 = \begin{pmatrix} 0 & 1 \\ 1 & 0 \end{pmatrix}~,
\gamma^2 = \begin{pmatrix} 1 & 0 \\ 0 & -1 \end{pmatrix}~.
\end{equation}
These satisfy the Clifford algebra $\{\gamma^{a},\gamma^{b}\} = -2 \eta^{ab}$, 
they are real and satisfy $\gamma^{012}=\gamma_{012} = 1 $. The Levi-Civita 
connection acting on spinors is
\begin{equation}
\nabla_{\mu} = \partial_{\mu}  -\frac14 \omega_{\mu ab } \gamma^{ab}~.
\end{equation}
For a real chiral spinor, we define complex coefficients as
\begin{equation}
\epsilon = \begin{pmatrix} \epsilon_1 \\ \epsilon_2 \end{pmatrix} 
\Longleftrightarrow \epsilon_z = \epsilon_1 + i\, \epsilon_2~.
\end{equation}
The complex coefficients have the property that Clifford multiplication by a two-dimensional
one-form corresponds to
\begin{equation}
\label{eq:cliffcomp}
\left( \phi_i \gamma^i \epsilon \right)_z = 2\, i\, e^{-\rho} \, \phi_z \, \epsilon_{\bar{z}}~.
\end{equation}
With the image of $\phi_z$ in the Clifford algebra $\phi_z(\gamma^1+i\gamma^2)$, the 
generator $L=-\frac12 \gamma^{12}$ acts on $\phi_z \mapsto -i \phi_z$ and $\epsilon_z \mapsto -\frac12 i \epsilon_z$. 
Equation \eqref{eq:cliffcomp} preserves the action as it should. The Levi-Civita connection becomes $\nabla_t\epsilon_z = \partial_t \epsilon_z$ and
\begin{align}
\nabla_z \epsilon_z &= \left( \partial_z -\frac12 \partial_z\rho\right)\epsilon_z~, \\
\nabla_{\bar{z}} \epsilon_z &= \left( \partial_{\bar{z}} + \frac12 \partial_{\bar{z}}\rho\right)\epsilon_z ~.
\end{align}

If we define the antisymmetric inner product by
\begin{equation}
( \epsilon, \epsilon' ) =\epsilon^T \gamma^0 \epsilon'~,
\end{equation}
then
\begin{align}
\epsilon_{\bar{z}} \epsilon'_z  &= - (\epsilon, \gamma^0 \, \epsilon') + i (\epsilon,\epsilon')\\
\epsilon_{\bar{z}} \epsilon'_{\bar{z}} &= (\epsilon, (\gamma^1 + i \gamma^2 ) \epsilon')~.\label{eq:epszz}
\end{align}
Requiring that two spinors $\epsilon_z$ and $\epsilon_z'$ do not square to a two-dimensional 
one-form is thus equivalent to $\epsilon_{\bar{z}} \epsilon'_{\bar{z}} =0$.

We now introduce our notation for chiral spinors in $S_{8+}$ of $\mathrm{Spin}(8)$. We define the Clifford algebra matrices
\begin{align}
\Gamma^I_{A\dot{A}} \Gamma^J_{B\dot{A}} + \Gamma^J_{A\dot{A}} \Gamma^I_{B\dot{A}} &= +2\delta^{IJ}\delta_{AB} \\
\Gamma^I_{A\dot{A}} \Gamma^J_{A\dot{B}} + \Gamma^J_{A\dot{A}} \Gamma^I_{A\dot{B}} &= +2\delta^{IJ}\delta_{\dot{A}\dot{B}}~
\end{align}
and
\begin{equation}
\Gamma^I = \begin{pmatrix} 0 & \Gamma^I_{B\dot{A}} \\ - \Gamma^I_{A\dot{B}} & 0
\end{pmatrix}
\end{equation}
acting on a non-chiral spinor $(\epsilon^A,\epsilon^{\dot{A}})\mapsto (\epsilon^B,\epsilon^{\dot{B}})$. 
The spin-invariant inner product is the identity matrix and the spin matrices $\Gamma^{IJ}_{AB} = -\Gamma^{[I}_{A\dot{A}}\Gamma^{J]}_{B\dot{A}}$ and  
$\Gamma^{IJ}_{\dot{A}\dot{B}} = -\Gamma^{[I}_{A\dot{A}}\Gamma^{J]}_{A\dot{B}}$ are antisymmetric with respect to the spin inner product and all matrices can be 
chosen to be real. The representation is chiral with $\Gamma^{12345678}=1$ on the real eight-dimensional spinors $\epsilon^A \in S_{8+}$.

By using these conventions we have the spin equivariant map from the square of real chiral spinors into the Clifford algebra
\begin{align}
S^2 S_{8+} &= \Lambda^0 \RR^8 \oplus \Lambda^4_+ \RR^8 \\
\Lambda^2 S_{8+} &= \Lambda^2 \RR^8~.
\end{align}
However, we are interested in complex chiral spinors in $S_{8+}^{\CC}$ that are isomorphic to the real tensor product of spacetime spinors with the spinors in $S_{8+}$,
\begin{equation}
\epsilon^A_z = \epsilon^A_1 + i \epsilon^A_2 \in S_{8+}^{\CC}~.
\end{equation}
In the above equation, $\epsilon^A_{\alpha}$ for $\alpha=1,2$ are the spin coefficients for each $A=1,\cdots, 8$ and our previous conventions apply. 

\subsection{Basis of timelike spinors}
\label{app:timespin}
In finding supersymmetric solutions of a theory, the form of the Killing spinor is usually fixed by using the symmetry of the theory. For our model
this was done in \cite{deger_supersymmetric_2010}. In this work we have instead used the symmetry $K$ to fix $\Pz^{Ir}$. Furthermore, we do not need to explicitly solve for the Killing spinors because the integrability of the gravitino variation is guaranteed in our analysis. Here we present a few complementary details on the Killing spinors once we have fixed $\Pz$ to a certain form. 

We choose a representation of the $\Gamma^{IJ}_{AB}$ matrices, such that the generators of the Cartan subalgebra $\Gamma^{12}_{AB}$,  $\Gamma^{34}_{AB}$,  
$\Gamma^{56}_{AB}$,  $\Gamma^{78}_{AB}$ are block diagonal and proportional to
\begin{equation}
\Gamma_{\sigma_1 \sigma_2 \sigma_3 \sigma_4} = 
\begin{pmatrix}
\sigma_1 \begin{pmatrix} 0 & 1 \\ -1 & 0 \end{pmatrix}
&&&\\
&
\sigma_2 \begin{pmatrix} 0 & 1 \\ -1 & 0 \end{pmatrix}
&&\\
&&
\sigma_3 \begin{pmatrix} 0 & 1 \\ -1 & 0 \end{pmatrix}
&\\
&&&
\sigma_4 \begin{pmatrix} 0 & 1 \\ -1 & 0 \end{pmatrix}
\end{pmatrix}~.
\end{equation}
This follows from Darboux's theorem, or equivalently because a two-form in $\SO(8)$ decomposes into a sum of $\Delta_0(\sigma_i,-\sigma_i)$ and $\Delta^{-}_{0}(0)$ in table \ref{tab:IndecOmn}. Since they need to 
square to $-1$, the $\sigma_i$ are all signs. The choice of which of the commuting $\Gamma^{IJ}_{AB}$ correspond to which of the $\Gamma_{\pm\pm\pm\pm}$ is 
restricted by the following rule: Any two products should trace to zero and the product of the four should be proportional to the identity. 

Up to reflections, there are only two choices for the signs $\sigma_i$. This is to be expected since the two chiral algebras are not isomorphic. We freely 
choose\footnote{The other choice is given by having only one (non-overlapping) signs different for each of the four generators.}
\begin{align}
\Gamma^{12} &= \Gamma_{++++}\\
\Gamma^{34} &= \Gamma_{++--}\\
\Gamma^{56} &= \Gamma_{+-+-}\\
\Gamma^{78} &= \Gamma_{+--+}~.
\end{align}

The condition $i\Gamma^{12}_{AB}\epsilon^B_z = \epsilon^A_z$ requires
\begin{equation}
\epsilon^A_z = \begin{pmatrix} i \epsilon_1 & \epsilon_1 
& i \epsilon_2 & \epsilon_2 & i \epsilon_3 & \epsilon_3
& i \epsilon_4 & \epsilon_4
\end{pmatrix}^T~.
\end{equation}
In fact, we can define a basis $\epsilon_{(\pm\pm\pm)}$ by
\begin{align}
i\Gamma^{12}\epsilon_{(\sigma_1 \sigma_2 \sigma_3)} &= \sigma_1 \epsilon_{(\sigma_1 \sigma_2 \sigma_3)} \\
i\Gamma^{34}\epsilon_{(\sigma_1 \sigma_2 \sigma_3)} &= \sigma_2 \epsilon_{(\sigma_1 \sigma_2 \sigma_3)} \\
i\Gamma^{56}\epsilon_{(\sigma_1 \sigma_2 \sigma_3)} &= \sigma_3 \epsilon_{(\sigma_1 \sigma_2 \sigma_3)} \\
i\Gamma^{78}\epsilon_{(\sigma_1 \sigma_2 \sigma_3)} &= \sigma_1\sigma_2\sigma_3 \epsilon_{(\sigma_1 \sigma_2 \sigma_3)} 
\end{align}
so that
\begin{align}
\epsilon_{(+++)} &=  \begin{pmatrix} i  & 1
&0&0&0&0&0&0
\end{pmatrix}^T
\\
\epsilon_{(++-)} &=  \begin{pmatrix}0&0& i  & 1
&0&0&0&0
\end{pmatrix}^T\\
\epsilon_{(+-+)} &=  \begin{pmatrix} 0&0&0&0& i  &1
&0&0
\end{pmatrix}^T
\\
\epsilon_{(+--)} &=  \begin{pmatrix}0&0&0&0&0&0& i & 1
\end{pmatrix}^T~.
\end{align}
The basis satisfies manifestly the condition $\epsilon^A_z \epsilon_z^A = 0$ so squaring any two timelike spinors, $(\epsilon^A,\gamma^{\mu} \epsilon'
{}^A)$ will be zero for components in the $\mu=1,2$ directions, see \eqref{eq:epszz}.

If a timelike background allows 8 real supersymmetries, the Killing spinors span the timelike spinor basis and we can fix a basis $\epsilon_{(i)z}^A$ such that 
each basis Killing spinor is proportional to one and only one of the $\epsilon^A_{(+\sigma_1\sigma_2)}$. The most general $N=8$ Killing spinor is
\begin{equation}
\epsilon_z^A = \sum_{\sigma_1,\sigma_2=\pm} F^{\sigma_1 \sigma_2} \epsilon^A_{(+\sigma_1 \sigma_2)}~,
\end{equation} 
where $F^{\sigma_1 \sigma_2}$ are functions of $z$.
If a timelike background allows 4 real supersymmetries, this arises from the algebraic supersymmetry equation (\ref{eq:algsusy})
imposing both 
$i\Gamma^{12}\epsilon_z = \epsilon_z$ and $i\Gamma^{34 }\epsilon_z = \epsilon_z$. The Killing spinor is now in the span of $\epsilon_{++\pm}$ and we can 
choose a basis of Killing spinors proportional to $\epsilon_{++\pm}$,
\begin{equation}
\epsilon_z^A = F^{+} \epsilon^A_{(+++)} + F^{-} \epsilon^A_{(++-)},
\end{equation}
where $F^{\pm}$ are functions of $z$.
Finally 2 real supersymmetries mean that there is a single basis Killing spinor proportional to $\epsilon_{+++}$. 

For $\tilde{n}=1,2,4$ complex supersymmetries
\begin{equation}
\epsilon^A_{(i)z} \quad (i)=1,\ldots, \tilde{n}~,
\end{equation}
there is an action of $\mathrm{SU}(\tilde{n})$ on the Killing spinors in the $\mathbb{R}$-linear span of the $\epsilon^A_{(i)z}$, which we now describe. First 
note that the matrix
\begin{equation}
\Delta_{(i)(j)} = \epsilon^A_{(i)z} \epsilon^A_{(j)z}
\end{equation}
is diagonal and constant. We can use a constant $\mathrm{GL}(\tilde{n},\CC)$ action $\Delta\mapsto M \Delta M^T$ in order to make it proportional to the identity. The 
matrices
\begin{equation}
M^{AB}_{(i)(j)} = \epsilon^A_{(i)z} \epsilon^B_{(j)\bar{z}}-  \epsilon^B_{(i)z} \epsilon^A_{(j)\bar{z}}~
\end{equation}
have some interesting properties. Since $(M^{AB}_{(i)(j)})^{*}=-M^{AB}_{(j)(i)}$ we have a map from $\mathfrak{su}(\tilde{n})$ into $\mathfrak{spin}(8)=\Lambda^2 S_{8+}$. 
For a constant $\mathfrak{su}(\tilde{n})$ matrix $(T_{ij})^{\dagger}=-T_{ij}$, the map is
\begin{equation}
T_{ij} \mapsto T^{AB} = T^{ij} M^{AB}_{ij}~.
\end{equation}
Indeed, the right hand side is real and antisymmetric in $A,B$. 
The group $\mathrm{SU}(\tilde{n})$ acts on the Killing spinor basis via spin rotations
\begin{equation}
T^{AB}  \epsilon_{(i) z}^B  = T^{ij} \epsilon_{(j)z}^A~.
\end{equation}
Under the $\mathrm{SU}(\tilde{n})$ we can essentially bring any timelike Killing spinor to be proportional to $\epsilon_{+++}$.

The $\SU(\tilde{n})$ action is important because we can make precise contact with other formulations. For instance, the $\mathfrak{so}(8)$ element $F^{IJ}$, which was 
called $\Omega^{IJ}$ in \cite{deger_supersymmetric_2010}, is given by the square of  $\epsilon_{+++}$:
\begin{equation}\label{eq:degerF}
F = e_1 \wedge e_2 + e_3 \wedge e_4 + e_5 \wedge e_6 + e_{7}\wedge e_8~.
\end{equation}
One can then find the eigenstates of $F^{IJ}$, which were called $\mathbb{P}^{Ii}$ in  \cite{deger_supersymmetric_2010}, and are simply the null basis
\begin{equation}\label{eq:degerP}
e_1+ i\, e_2~,~~e_3+i\,e_4~,~~e_5+i\,e_6~,~~e_7+i\,e_8~.
\end{equation}
We thus understand that the result of \cite{deger_supersymmetric_2010}, that $\Pz$ should be expanded in  $\mathbb{P}^{Ii}$, is equivalent to our complex null 
basis $\cb{e}{i}$ of the main text. 

\section{Direct matrix factorizations of $\Pz$}
\label{app:glnc}
We give here a direct analysis of how the form of $\Pz$ can be fixed if we use $K^{\CC}$ or the maximal symmetry of \eqref{eq:algsusy}, namely $\mathrm{SO}(8)^{\CC}\times \mathrm{GL}(n,\CC)$, without the nilpotency argument. This gives an alternative proof that the real supersymmetries come in powers of two. We begin with two useful lemmas.

\begin{lemma}\label{lem:uso}
The subgroups $\mathrm{U}(1)^\mu$ and $\mathrm{SO}(\mu)$ of 
$\mathrm{GL}(\mu,\CC)$ generate $\mathrm{U}(\mu)$.
\end{lemma}
\begin{proof}
Consider a complex orthonormal basis $\{e_i\}_{i=1}^{\mu}$ of $\CC^\mu$ with respect to the Hermitian inner product on $\CC^{\mu}$ and its 
Hermitian dual $\{(e_i)^{\flat}\}_{i=1}^{\mu}$. The generators $L_{ij}$ of $\mathfrak{so}(\mu)$ are
\begin{equation}
L_{ij} = e^i \otimes (e^j)^{\flat} -  e^j \otimes (e^i)^{\flat} \, \qquad i\neq j
\end{equation}
and the $\mathfrak{u}(1)^{\mu}$ generators are 
\begin{equation}
L_i=i\,e^i \otimes (e^i)^{\flat} \qquad \text{(no sum over $i$)}~.
\end{equation}Their commutator is
\begin{equation}\label{eq:gensym}
[L_i,L_{ij}]= i\left( e^i \otimes  (e^j)^{\flat} + e^j \otimes (e^i)^{\flat}\right)~.
\end{equation}
All of the generators of $\mathrm{SU}(\mu)$ are thus generated from the group product $\mathrm{U}(1)^\mu\mathrm{SO}(\mu)$. On the other hand, the group generated 
preserves the Hermitian inner product on $\CC^\mu$ so it cannot be larger than $\mathrm{U}(\mu)$. Finally, we can assert that the group contains the non-special unitary 
$\mathrm{U}(1)$ and is fully $\mathrm{U}(\mu)$.
\end{proof}
\begin{lemma}\label{lem:glnc}
The groups $\left(\RR^+ \times \mathrm{U}(1)\right)^\mu$ and $\mathrm{SO}(\mu)^{\CC}$ in $\mathrm{GL}(\mu,\CC)$ generate $\mathrm{GL}(\mu,\CC)$.
\end{lemma}
\begin{proof}
The scaling $\RR^+$ is given by $i \, L_i$, where we use $L_i$ and $L_{ij}$ of lemma \ref{lem:uso}. Clearly, all matrices in $\mathrm{GL}(\mu,\CC)$ can now be 
generated (symmetric and antisymmetric, real and imaginary) similarly to \eqref{eq:gensym}.
\end{proof}

Assume an element $\Pz$ that admits some supersymmetry according to the algebraic supersymmetry equation (\ref{eq:algsusy}). 
Multiplying the equation with $\Pz^{Is}\Gamma^I$ and symmetrizing over $(r,s)$ we derive 
\begin{equation}
\Pz^{Ir}\Pz^{Is} = 0~.
\end{equation}
By using $\mathrm{SO}(8)^{\CC}$, we fix it as
\begin{equation}\label{eq:Ptildeso8c}
P^{Ir}_z e_I \otimes \hat{e}_r \conjugate{\mathrm{SO}(8)^{\CC}} \sum_{i=1}^4\tilde{P}^{ir}
\left(e_{2i-1}+i\, \text{sign}_i \, e_{2i}\right) \otimes \hat{e}_r ~,
\end{equation}
with the group 
\begin{equation}
\left(\RR^+ \times \mathrm{U}(1)\right)^4 \cdot
\mathrm{SO}(4)^{\CC} = \mathrm{GL}(4,\CC)~ \label{eq:rscalinggl}
\end{equation}
acting on the left. The equality in \eqref{eq:rscalinggl} follows from lemma \ref{lem:glnc}. The factors of $\mathrm{U}(1)$ come from the rotation $e_{2i-1}+i\,e_{2i}\mapsto i\left(e_{2i-1}+i\,e_{2i}\right)$, the $\SO(4)^{\CC}$ is manifestly a 
subgroup of $\SO(8)^{\CC}$, and the scalings $\RR^+$ are the complex $\SO(2)^{\CC}$ rotations that are not in $\SO(2)$. 

The element $\Pz$ is represented by a $4\times n$ matrix $\tilde{P}^{ir}$ in \eqref{eq:Ptildeso8c} and inherited from $\SO(8)^{\CC}\times \SO(n)^{\CC}$ is the group  
$\mathrm{GL}(4,\CC)\times \SO(n)^{\CC}$ acting on $\tilde{P}^{ir}$ by left/right multiplication. Similarly, the group inherited from 
$\SO(8)^{\CC}\times \mathrm{GL}(n,\CC)$ acting on $\tilde{P}^{ir}$ is $\mathrm{GL}(4,\CC)\times \mathrm{GL}(n,\CC)$. The rank of $\tilde{P}^{ir}$ is not necessarily 
full. 

The action of $\mathrm{GL}(4,\CC)$ and permutations in $\mathrm{SO}(n)^{\CC}$ can be used to rotate the $\tilde{P}^{ir}$
to one of the following forms
\begin{multline}
\begin{pmatrix}
0\cdots 0 \\
0\cdots 0\\
0\cdots 0\\
0\cdots 0
\end{pmatrix}~, 
\begin{pmatrix}
1 & *\cdots * \\
0 & 0\cdots 0\\
0 & 0\cdots 0\\
0 & 0\cdots 0
\end{pmatrix}~,
\begin{pmatrix}
1 & 0 & *\cdots * \\
0 & 1 & *\cdots *\\
0 & 0 & 0\cdots 0 \\
0 & 0 & 0\cdots 0
\end{pmatrix}~,
\begin{pmatrix}
1 & 0 & 0& *\cdots *\\
0 & 1 & 0& *\cdots *\\
0 & 0 & 1& *\cdots *\\
0 & 0 & 0& 0\cdots 0
\end{pmatrix}~,\\
\begin{pmatrix}
1 & 0 & 0& 0 &*\cdots *\\
0 & 1 & 0& 0 &*\cdots *\\
0 & 0 & 1& 0 &* \cdots *\\
0 & 0 & 0& 1 &* \cdots*
\end{pmatrix}
\label{eq:Glfixedmatrices}
\end{multline}
for respectively $16,8,4,2,2$ real supersymmetries. Stars signify here
possibly non-zero elements. The reason why the upper-left
square block is the identity matrix rather than
a triangular matrix is because of the action of the stabilizers of
one, two, three and four complex vectors in $\mathrm{GL}(4,\CC)$:
\begin{equation}
\begin{aligned}
\mathrm{GL}(4,\CC) \supset & \mathrm{GL}(3,\CC) \ltimes \RR^3 \\
\supset& \mathrm{GL}(2,\CC) \ltimes \left( \RR^2 \oplus \RR^2 \right) \\
\supset& \mathrm{GL}(1,\CC) \ltimes \left( \RR \oplus \RR \oplus \RR
\right)\\
\supset & 1~.
\end{aligned}
\end{equation}
In particular, the $p$ copies of $\RR^{4-p}$ (for $p=1,2,3$) are translations that set
the first $p$ components of the next column to be fixed (the
$(p+1)$'th column) equal to zero. 

The matrices in \eqref{eq:Glfixedmatrices} describe the QR decomposition of $\tilde{P}^{ir}$ with respect to $\mathrm{GL}(4,\CC)$ acting on the left. On the other hand, the 
group $\SO(8)^{\CC}\times \mathrm{GL}(n,\CC)$ is such that all stars in \eqref{eq:Glfixedmatrices} may be fixed to zero. The classification under  
$\SO(8)^{\CC}\times \mathrm{GL}(n,\CC)$, the maximal symmetry of the supersymmetry equation, thus describes finite classes with each class representing 
uniquely a certain fraction of supersymmetry. Whichever of these two groups we use, or indeed if we use the factorization of $\Pz$ under $K$ with a similar method to 
the above, the real supersymmetry can be shown to come in powers of two.
\section{Supersymmetry in the Zariski topology}
\label{app:zariski}
We review here the result 2(a) of \cite{de_boer_classifying_2014}. The proof is identical with minor changes. More precisely, we prove the statement ``if an orbit $O$ of an element $\Pz$ admits (at least) $\tilde{n}$ supersymmetries, then elements in the closure $\bar{O}$ in the Zariski topology preserve at least $\tilde{n}$ supersymmetries''. The Zariski topology is defined in terms of its closed sets. A closed set in the Zariski topology on a space $\mathfrak{M}$ (in this case $\mathfrak{M}=\mathfrak{g}^{\CC}$ is a complex Lie algebra) is by definition the solution space of a finite set of homogeneous polynomial equations on $\mathfrak{M}$. That is, $V(S)$ is a closed set if
\[ V(S) = \{ X \in \mathfrak{M} : \, f(X)=0 \,\,\, \forall f \in S \}~,\]
where $S$ is an ideal of homogeneous polynomials on $\mathfrak{M}$.

Let us first consider all elements $\Pz$ that preserve at least $\tilde{n}$ \emph{complex} supersymmetries. The condition is that
\begin{equation}\label{eq:algsusytop}
\Pz^{Ir}\Gamma^I_{A\dot{A}}\epsilon_{\bar{z}}^A=0~,
\end{equation}
for at least $\tilde{n}$ linearly independent spinors $\epsilon_z^A$. Via the rank-nullity theorem the rank of the $8n \times 8$ matrix $\Pz^{Ir}\Gamma^I_{A\dot{A}}$ (acting on the left of $\epsilon^A_{\bar{z}}$) is 
at most  $(8-\tilde{n})$. All $(9-\tilde{n})\times (9-\tilde{n})$ submatrices of $\Pz^{Ir}\Gamma^I_{A\dot{A}}$ should thus have vanishing determinant. The elements we 
are considering are evidently roots of a finite number of homogeneous polynomial equations. The condition in \eqref{eq:algsusytop} that $P_z$ admits at least $\tilde{n}$ supersymmetries is seen to be equivalent to the condition that $P_z\in C_{\tilde{n}}$, where $C_{\tilde{n}}$ is the solution space of certain homogeneous polynomial equations of degree  $9-\tilde{n}$ in the components $P^{Ir}_z$. In particular, if the root space of the determinant of a certain $(9-\tilde{n})\times (9-\tilde{n})$ submatrix is $D_i$ then $C_{\tilde{n}}=\cap_i D_i$ with the index $i$ running over all such submatrices. Let us add a comment here. If the element $P_z$ preserves precisely $\tilde{n}$ supersymmetries, then $P_z \in C_{\tilde{n}}$ as well as $P_z \in C_{\tilde{n}'}$ for all $\tilde{n}'\leq \tilde{n}$ but $P_z \notin C_{\tilde{n}'}$ for $\tilde{n}'> \tilde{n}$. Indeed, if all  $(9-\tilde{n})\times (9-\tilde{n})$ submatrix determinants of  $\Pz^{Ir}\Gamma^I_{A\dot{A}}$ are zero, the determinants of bigger size submatrices will also be zero. We have the partial ordering
\[ C_{\tilde{m}} \subseteq C_{\tilde{n}} \text{ if and only if }\tilde{m} \geq \tilde{n}~.\]
Additionally, $P_z$ preserves at least $\tilde{n}$ supersymmetries if and only if $P_z \in C_{\tilde{n}}$. In our argument we assume that $\Pz$ preserves at least $\tilde{n}$ supersymmetries but can be made stricter by assuming precisely $\tilde{n}$ supersymmetries. We do not gain any advantage with the stricter assumption. Since the $C_{\tilde{n}}$ are defined in terms of homogeneous polynomials, we may assert that the $C_{\tilde{n}}$ are closed in the Zariski topology on $\mathfrak{g}^{\CC}$.

Let us take the orbit $O$ under conjugacy by $K^{\CC}$ of an element $P_z \in C_{\tilde{n}}$. Since the action of $K^{\CC}$ preserves supersymmetry, we may assert that $O\subseteq C_{\tilde{n}}$. The closure $\bar{O}$ of the orbit $O$ should be a subset of $C_{\tilde{n}}$ as well. Indeed $\bar{O}\subseteq C_{\tilde{n}}$ follows due to closure: any sequence in $O$, which is contained in $C_{\tilde{n}}$, is also a sequence in the already closed set $C_{\tilde{n}}$. Now take any other orbit $O'\subseteq \bar{O}$ of some element $P_z'$. It should evidently satisfy $O'\subseteq \bar{O} \subseteq C_{\tilde{n}}$. Therefore the orbit $O'$ and $P_z' \in O'$ admit at least $\tilde{n}$ supersymmetries. We have proven the original statement of result 2(a) on page 21 of \cite{de_boer_classifying_2014}: ``if an orbit $O$ of an element $\Pz$ admits (at least) $\tilde{n}$ supersymmetries, then elements in the closure $\bar{O}$ in the Zariski topology preserve at least $\tilde{n}$ supersymmetries''. 

Let us comment that different elements of $\bar{O}$ might preserve in principal \emph{different} amounts of supersymmetry, so we refrain from saying ``$\bar{O}$ preserves at least the same amount of supersymmetry as $O$''. Let us also remark the power of turning towards the Zariski topology. It allows us to use the theorem by Kostant and Rallis,  lemma 11 in \cite{kostant_orbits_1971}, that the closure of an orbit of a non-nilpotent element in the Zariski topology necessarily contains a semi-simple element. But as in \cite{de_boer_classifying_2014}, we show in the main text that semi-simple elements preserve no supersymmetry and hence all supersymmetric elements $P_z$ are nilpotent. 

\section{Constructing normal forms}
\label{app:constnormalforms}
\subsection{Kostant-Segikuchi correspondence}
\label{sec:KSC}
Let us define $\theta_C$ the
Cartan involution of a real Lie algebra $\mathfrak{g}$. That is, the algebra decomposes as
\begin{align*}
  \mathfrak{g}&= \mathfrak{k}\oplus \mathfrak{p}\\
  \theta_C&= \left.+1\right.|_{\mathfrak{k}}\oplus
  \left.-1\right.|_{\mathfrak{p}} ~,
\end{align*}
where $\mathfrak{k}$ is the maximally compact subalgebra. We will eventually take
$\mathfrak{g}= \mathfrak{so}(8,n)$ and $\mathfrak{k}=
\mathfrak{so}(8)\oplus \mathfrak{so}(n)$. 

A standard triple $\{E,F,H\}$ is an ordered set of elements in
$\mathfrak{g}$ or $\mathfrak{g}^{\CC}$ (depending on the context) that
generate $\mathfrak{sl}_2$ and with canonical relations 
\begin{equation}
[H,E]=2 E, \quad 
[H,F]=-2 F \text{ and }[E,F]=H.
\end{equation} 
We define a Kostant-Segikuchi triple $\{E,F,H\}$ in
$\mathfrak{g}$ to be a standard triple such that
\begin{equation}
\label{eq:KSTgCondition}
  F = - \theta_C E~.
\end{equation}
From this it also follows that $\theta_C H = -H$. We also define a Kostant-Segikuchi triple 
$\{e,f,h\}$ in
$\mathfrak{g}^{\CC}$ to be a standard triple such that 
\begin{equation}
 \begin{aligned}
 f  &= e^{*} \\
 \theta_C e &= - e~.
 \end{aligned}
\end{equation}
From this it also follows that $\theta_C h = h$. 

The Kostant-Segikuchi correspondence establishes the correspondence
between Kostant-Segikuchi triples in $\mathfrak{g}$ up to the action of $G$ and Kostant-Segikuchi triples in
$\mathfrak{p}^{\CC}$ up to the action of $K^{\CC}$. By an adaptation of the 
Jacobson-Morozov theorem, this is a correspondence between nilpotent
elements in $\mathfrak{g}$ up to the action of $G$ and nilpotent elements in $\mathfrak{p}^{\CC}$ up to the action of $K^{\CC}$:
\begin{equation}
 \text{Nil}[ \mathfrak{g} ] / G  = \text{Nil}[ \mathfrak{p}^{\CC}]/ K^{\CC}~.
\end{equation}
Explicitly, the correspondence is given by
\begin{align}
e&=\frac12(E+F+iH)\\
f&=\frac12(E+F-iH)\\
h&=i(E-F)~.
\end{align}

\subsection{Normal forms in $\mathfrak{g}$}
\label{app:normalforms}
Indecomposable types can be classified as follows: Let
\begin{equation}
A = S + N
\end{equation}
be the Jordan-Chevalley decomposition corresponding to an indecomposable element $A\in L(V,\tau,\sigma)$ and $(A,V)\in \Delta$. The definition of $\Delta$ was given in section \ref{sec:indecomposable}. Let the order of the nilpotent part $N$ be 
$p$, that is $N^{p+1}=0$ in the fundamental representation\footnote{This is not necessarily the same as the order in the adjoint.}. By proposition 
3 in \cite{burgoyne_conjugacy_1977}, it is true that $\text{Ker}N^m=NV$. We define the non-degenerate form $\bar{\tau}$ on $\bar{V}=V/NV$ as 
\begin{equation}
\label{eq:taubardef}
\bar{\tau}(u,v) = \tau(u,N^p v)~,
\end{equation}
which has symmetry $|\tau|(-1)^p$ where $|\tau|$ is the symmetry of $\tau$. By proposition 3 again, the restriction $\bar{A}$ of $S$ acting on $\bar{V}$ is well-defined, 
semisimple and $(\bar{A},\bar{V})\in \bar{\Delta}$ is an indecomposable type of $L(\bar{V},\bar{\tau},\bar{\sigma})$. Proposition 2 in 
\cite{burgoyne_conjugacy_1977} asserts that
\begin{theorem} 
\label{thm:indecDdetermined}
An indecomposable type $\Delta$ is completely determined by $p$ and $\bar{\Delta}$. 
\end{theorem}

According to theorem \ref{thm:indecDdetermined}, 
in order to classify indecomposable types $\Delta$, what remains is to classify the indecomposable 
semisimple types $\bar{\Delta}$ of certain linear algebras $L(\bar{V},\bar{\tau},\bar{\sigma})$. 
The semisimple types are labeled by their eigenvalues $(\zeta,\cdots)$ on $\bar{V}$. We refer to 
\cite{burgoyne_conjugacy_1977} for further details and for a proof of the multiplicity of the eigenvalues. For the problem at hand, we have listed the indecomposable 
types of $\mathrm{O}(m,n)$ in table \ref{tab:IndecOmn}. In particular, we are interested in the nilpotent elements given in the last two rows of the table. 

Although \cite{burgoyne_conjugacy_1977} does not list explicit normal forms, these can be easily constructed based on the proof of 
proposition 2 in \cite{burgoyne_conjugacy_1977} that extends lemma 2 in \cite{burgoyne_conjugacy_1977}:
\begin{theorem}
\label{thm:indecDConstruction}
Suppose $A\in L(V,\tau,\sigma)$ is such that its nilpotent part $N$ has order $p$ 
and $NV = \text{ker} N^p$. Then there exists an $S$-invariant and $\sigma$-invariant subspace $W$ such that 
\begin{equation}
V=W\oplus NW \oplus \cdots \oplus N^p W
\end{equation} 
is a sum of mutually disjoint subspaces with the following properties
\begin{itemize}
\item $W=\bar{V}$ as a complement of $NV$ in $V$,
\item $\dim N^i H = \dim H$ for $0\leq i \leq p$,
\item $\tau(u,N^i v)=0$ for $u,v\in W$ and $0\leq i < p$.
\end{itemize}
\end{theorem}

The conditions of theorem \ref{thm:indecDConstruction} are met for elements of an indecomposable type. 
It follows from theorem \ref{thm:indecDConstruction} that for two elements 
$u=\sum_i N^i u_i \in V$ and $v=\sum_iN^i v_i\in V$, their inner product is determined by that on $W$
\begin{equation}
\label{eq:taured}
\tau(u,v) = \sum_{i+j=p}(-1)^i \bar{\tau}(u_i,v_j)~.
\end{equation}
We remind the reader that the symmetry of $\bar{\tau}$ now also depends on $p \mod 2$. Assume then that we have identified the space $W\subset V$ and that we specify 
the irreducible type $\bar{\Delta}$ of $\bar{A}$ that is $S$ acting on $W=\bar{V}$ as an operator in $L(W,\bar{\tau},\sigma)$. The normal form of $A=S+N$ acting on $V$ can be 
constructed as follows:
\begin{enumerate}
\item the operator $N$ is the ladder operation on  $V=\oplus_i N^i W$. It is left undetermined up to scalings of each ladder-step operation.
\item The normal form of $S$ is given by extending\footnote{Recall that $S$ and $N$ commute by the Jordan-Chevalley theorem.} the action of $\bar{A}$ from $W$ to $V$.
\end{enumerate}   
This is essentially the method we will use to write normal forms. That is, we identify $W$ in an explicit basis $V$ and construct $S$ and $N$ accordingly.

We define appropriately a basis of $V$ with the requisite signature of table \ref{tab:IndecOmn} and inspect the left-hand side of \eqref{eq:taured}. This allows us to 
identify the subspace $W$ such
that $\tau(\phi (\cdot) , \phi (\cdot) )$ is non-zero only for 
\begin{equation}
\tau(N^i u, N^{p-i}v)\text{ with }u,v \in W~.
\end{equation} 
It is trivial to write the nilpotent part as the ladder operators
$N^iW\rightarrow N^{i+1}W$ and the semisimple part as the operator
with the requisite eigenvalues on $V$. We are interested in nilpotent elements so $S=0$ and the dimension of $W$ is given by the multiplicity of zeros in the notation 
of table \ref{tab:IndecOmn}: One and two for $\Delta^{\pm}_p(0)$ and $\Delta_p(0,0)$, respectively. 
In the following two subsections, we give normal forms for these elements 
only. Nevertheless, one can easily use this method to find a normal form for any type in the table.

\subsubsection{Type $\Delta_p^{\pm}(0)$}
\label{app:typeeing}
We consider $\Delta_p^{\pm}(0)$ with $p\in2\mathbb{N}$  on a vector space $V$ of signature
$\pm(-1)^{\frac{p}2}(\frac{p}2+1,\frac{p}2)$. It is generated by $v$,
$\sigma v=\pm v$, $\bar{\tau}(v,v)=1$. Depending on the sign of the
real structure, we take $\tilde{v}=v$ or $\tilde{v}=i v$ such that it
is real and $\tau(\tilde{v},\tilde{v})=\pm 1$.

All elements of $V$ are
of the form $N^i\tilde{v}$, $i=0,1,\cdots,p$. The inner product is
\begin{align*}
\tau(N^k\tilde{v},N^{p-k}\tilde{v})&=\pm(-1)^k&&k=0,1,\cdots,\frac{p}{2}-1\\
  \tau(N^{\frac{p}2}\tilde{v},N^{\frac{p}2}\tilde{v})&=\pm(-1)^{
    \frac{p}2} ~.&&
\end{align*}
We choose the null basis
\begin{equation}
\{\eta^k,\tilde{\eta}^k,\theta\}
\end{equation}
with $\eta^k=N^k\tilde{v}$,
$\tilde{\eta}^k=\pm(-1)^{k}N^{p-k}\tilde{v}$,
$\theta=N^{\frac{p}2}\tilde{v}$. The inner product is thus non-zero on
\begin{align}
\tau( \eta^i, \tilde{\eta}^j) &= \delta^{ij}\\
\tau( \theta, \theta) &=  \pm(-1)^{\frac{p}2}~.
\end{align}
One can construct $N$ using the fact that is a ladder operator
\begin{equation}
\label{eq:normaltypeeing} 
 N = \sum_{k=0}^{\frac{p}2-2} a_k \eta^{k+1} \wedge \tilde{\eta}^k
  +b\, \theta\wedge \tilde{\eta}^{\frac{p}2-1}~.
\end{equation}
The coefficient can be scaled freely. We will later choose appropriately so that $N$ belongs to a KS triple.

\subsubsection{Type $\Delta_p(0,0)$}
The case is identical to case $\Delta_p(\zeta,-\zeta)$ with $S=0$. We consider $\Delta_p(0,0)$ with $p\in2\mathbb{N}+1$ on a vector space $V$ of signature $(p+1,p+1)$. It is
generated by the highest-weight vectors $v$ and $w$, $\sigma v=i\,w$, $\sigma w = i
\, v$ and $\bar{\tau}(v,w)=1$. We take the real $\tilde{v}=(v+\sigma
v)/\sqrt{2}$ and $\tilde{w}=i(v-\sigma v)/\sqrt{2}$ with
$\bar{\tau}(\tilde{v},\tilde{w})=1$.

All elements of $V$ are of the form $N^i\tilde{v}$ and
$N^i\tilde{w}$, $i=0,1,\cdots,p$. The inner product is
\begin{eqnarray*}
\left.
\begin{array}{ll}
\tau(N^k\tilde{v},N^{p-k}\tilde{w}) & = (-1)^k \\
\tau(N^k\tilde{w},N^{p-k}\tilde{v}) & =(-1)^{k+1}
\end{array}
\right\} \quad k=0,1,\cdots,\frac{p-1}{2} \,.
\end{eqnarray*}
We choose the null basis
\[\{\eta^k,\tilde{\eta}^k,\theta^k,\tilde{\theta}^k\} \,. \] With
$\eta^k=N^k\tilde{v}$, $\tilde{\eta}^k=(-1)^{k}N^{p-k}\tilde{w}$,
$\theta^k=N^k\tilde{w}$,
$\tilde{\theta}^{k}=(-1)^{k+1}N^{p-k}\tilde{v}$. The non-zero inner product is
\begin{align}
\tau( \eta^i, \tilde{\eta}^j) & = \delta^{ij}~,\\
\tau( \theta^i, \tilde{\theta}^j ) &= \delta^{ij}~.
\end{align}
$N$ is the ladder operator
\begin{equation*}
  N = \sum_{k=0}^{\frac{p-1}2-1} \left(a_k\, \eta^{k+1} \wedge \tilde{\eta}^k + 
  b_k\, \theta^{k+1}\wedge \tilde{\theta}^k \right)+ c\, \tilde{\eta}^{\frac{p-1}2}\wedge \tilde{\theta}^{\frac{p-1}2}~,
\end{equation*}
where the $a_k$, $b_k$ and $c$ are constants that can be scaled freely.

\subsection{Kostant-Segikuchi triples in $\mathfrak{g}$}
\label{app:KST}
In order to construct normal forms for elements in $\mathfrak{p}^{\CC}$ up to the action of $\mathfrak{k}^{\CC}$, we will use the Kostant-Segikuchi correspondence. We are thus interested in Kostant-Segikuchi triples in $\mathfrak{so}(8,n)$. That is, we need to construct triples of the form $\{E,F,H\}$ that satisfy the condition $F = - \theta_C E$.

The construction in appendix \ref{app:normalforms} used the simplest coefficients for a nilpotent part of an element $N$. By using boosts, we amend the normal form of a nilpotent 
 element $E$ such that $E$, $F=-\theta_C E$ and
\begin{equation}
H= [E,F]=-[N,\theta_C N]
\end{equation}
indeed satisfy the standard $\mathfrak{sl}_2$ relations. This can always be done and it fixes the scalings of the ladder operators
completely. 

We give the normal form of Kostant-Segikuchi triples here and using the Kostant-Segikuchi correspondence we give the corresponding nilpotent element in $\mathfrak{p}^{\CC}$ 
in the subsection \ref{app:norminpc}. There are two nilpotent complex types in $\mathfrak{p}^{\CC}$: One inherited from the indecomposable type $\Delta_p^{\pm}(0)$ of signature 
$\pm(-1)^{\frac{p}2}(\frac{p}2+1,\frac{p}2)$ with $p$ even, and one from the indecomposable $\Delta_p(0,0)$ of signature $(p+1,p+1)$ with $p$ odd. 

\subsubsection{Type  $\Delta_p^{\pm}(0)$}
Recall that type $\Delta_p^{\pm}(0)$ with $p\in2\mathbb{N}$ is of signature $\pm(-1)^{\frac{p}2}(\frac{p}2+1,\frac{p}2)$
in Table \eqref{tab:IndecOmn}. We use the
basis $\{\eta^k,\tilde{\eta}^k,\theta\}$ as before and the nilpotent normal form $N$ is
\begin{align*}
  E &= \sum_{i=0}^{\frac{p}2-2} a_i \eta^{i+1} \wedge \tilde{\eta}^i
  +b \, \theta \wedge \tilde{\eta}^{\frac{p}2-1}~,
\end{align*}
where the $a_i$ and $b$ are to be determined. By using 
\begin{align}
\theta_C(\eta^i\wedge \tilde{\eta}^j) &=\tilde{\eta}^i\wedge \eta^j
\\\intertext{and}
\theta_C(\theta\wedge \tilde{\eta}^i) &= \pm(-1)^{\frac{p}2} \,\theta\wedge\eta^i~,
\end{align}
we compute
 \begin{equation}
  F =-\theta_C N= \sum_{i=0}^{\frac{p}2-2} a_i \eta^i\wedge
  \tilde{\eta}^{i+1}+ \pm(-1)^{\frac{p}2} \, b \, \eta^{\frac{p}2-1}\wedge \theta~.
  \end{equation}
We need to impose that $E$ and $F$ form part of a Kostant-Segikuchi triple.

Consider the action of $E$ and $F$ on the basis of $\mathbb{R}^{\frac{p}2+1,\frac{p}2}$ if $\pm(-1)^{\frac{p}2}=1$ and $\mathbb{R}^{\frac{p}2,\frac{p}2+1}$ if $\pm(-1)^{\frac{p}2}=-1$:
\begin{equation}
\begin{array}{rllrll}
E: & 
\eta^i & \mapsto 
a_i \, \eta^{i+1} \quad
&
F: &
\eta^{i+1} & \mapsto
a_i \, \eta^i
\\
E: & 
\eta^{\frac{p}2-1} & \mapsto 
b \, \theta \quad
&
F: &
\theta & \mapsto
b \, \eta^{\frac{p}2-1}
\\
E: & 
\theta & \mapsto 
-\pm(-1)^{\frac{p}2}\,b\, \, \tilde{\eta}^{\frac{p}2-1} \quad
&
F :&
\tilde{\eta}^{\frac{p}2-1} & \mapsto 
-b\,\pm(-1)^{\frac{p}2}\, \theta \\
E: & 
\tilde{\eta}^{i+1} & \mapsto 
-a_i \, \tilde{\eta}^i \quad
&
F: &
\tilde{\eta}^i & \mapsto
- a_i\, \tilde{\eta}^{i+1}
\end{array}
\end{equation}
where $i=0,\cdots,\frac{p}2-2$, and we define for consistency $a_{-1}=0$. The $\mathfrak{sl}_2$ algebra requires
\begin{align}
H\eta^i & = (2i-p) \eta^i~, &  i&=0,\cdots ,\frac{p}2-1 \\
H \tilde{\eta}^i &= (p-2 i) \tilde{\eta}^i ~,&  i&=0,\cdots ,\frac{p}2-1 \\
H \theta &= 0 ~. && 
\end{align}
A straightforward calculation gives us the conditions
\begin{align}
a_i^2 - a_{i-1}^2 &= p - 2 i~, & i&=0,\cdots, \frac{p}2-2
\\
a\left(\frac{p}2-2\right)^2 &= b^2-2 ~.&&
\end{align}
The constraint determines all constants uniquely
\begin{align*}
  a_i^2 &= (p-i)(i+1) ~,\quad i=0,\cdots ,\frac{p}2-2\\
  b^2 &=\frac{p}2 \left( \frac{p}2 +1\right)
\end{align*}
and the (hyperbolic) element $H$ is
\begin{equation}
H = \sum_{i=0}^{\frac{p}2-1}\left(p- 2i\right)\tilde{\eta}^i \wedge \eta^i~.
\end{equation}
Equation \eqref{eq:einpctypee} is given by using the Kostant-Segikuchi correspondence and switching to an orthonormal frame
\begin{align}
e_i &= \frac{\sqrt{2}}{{2}} \left( \eta_i + \tilde{\eta}_i \right)~,\\
\hat{e}_i &= \frac{\sqrt{2}}{{2}} \left( \eta_i - \tilde{\eta}_i \right) ~.
\end{align}

\subsubsection{Type  $\Delta_p(0,0)$}
Recall that type $\Delta_p(0,0)$ with $p\in2\mathbb{N}+1$ is of
signature $(p+1,p+1)$ and we use the basis
$\{\eta^k,\tilde{\eta}^k,\theta^k,\tilde{\theta}^k\}$ of $\mathbb{R}^{p+1,p+1}$, $k=0,1,\cdots
\frac{p-1}2$. Previously, we had used the nilpotent normal form
\begin{equation}
  N = \sum_{i=0}^{\frac{p-1}2-1}\left(\eta^{i+1} \wedge
    \tilde{\eta}^i +\theta^{i+1}\wedge \tilde{\theta}^i\right) +
  \tilde{\eta}^{\frac{p-1}2} \wedge \tilde{\theta}^{\frac{p-1}2}~.
\end{equation} 
We boost $N$ and use the nilpotent element
\begin{equation}
  N = \sum_{i=0}^{\frac{p-1}2-1}\left( a_i \eta^{i+1} \wedge
    \tilde{\eta}^i + b_i \theta^{i+1}\wedge \tilde{\theta}^i\right) +
  c \, \tilde{\eta}^{\frac{p-1}2} \wedge \tilde{\theta}^{\frac{p-1}2}~,
\end{equation}
where the $a_i$, $b_i$ and $c$ are constants to be determined. Note that there is still a manifest SO(1,1) freedom. 
We calculate
\begin{equation}
  F =-\theta_C N = \sum_{i=0}^{\frac{p-1}2-1}\left(a_i \, \eta^{i}
    \wedge \tilde{\eta}^{i+1} +b_i\,\theta^{i }\wedge
    \tilde{\eta}^{i+1}\right) - c\, {\eta}^{\frac{p-1}2} \wedge
  {\theta}^{\frac{p-1}2}~.
\end{equation}
As before, we will impose that these two form the parabolic parts of a standard triple.

Let us write the action of $E$ and $F$ on the basis. It is
\begin{equation}
\begin{array}{rllrll}
E: & 
\eta^i & \mapsto 
a_i \, \eta^{i+1} \quad
&
F: &
\eta^{i+1} & \mapsto
a_i \, \eta^i
\\
E: & 
\eta^{\frac{p-1}2} & \mapsto 
-c \, \tilde{\theta}^{\frac{p-1}2} \quad
&
F: &
\tilde{\theta}^{\frac{p-1}2} & \mapsto
-c \, \eta^{\frac{p-1}2}
\\
E: & 
\tilde{\theta}^{i+1} & \mapsto 
-b_i \, \tilde{\theta}^{i} \quad
&
F :&
\tilde{\theta}^i & \mapsto 
-b_i\, \tilde{\theta}^{i+1} \\
E: & 
 \theta^i & \mapsto 
 b_i \, \theta^{i+1} \quad
&
F: &
\theta^{i+1} & \mapsto
b_i\, \theta^i \\
E: & 
\theta^{\frac{p-1}2} & \mapsto 
c \, \tilde{\eta}^{\frac{p-1}2} \quad
&
F: &
\tilde{\eta}^{\frac{p-1}2} & \mapsto
c\, \theta^{\frac{p-1}2}\\
E: & 
\tilde{\eta}^{i+1} & \mapsto 
-a_i \, \tilde{\eta}^i \quad
&
F: &
\tilde{\eta}^i & \mapsto
- a_i\, \tilde{\eta}^{i+1}
\end{array} ~,
\end{equation}
where $i=0,\cdots,\frac{p-1}2-1$. We may also put $a_{-1}=b_{-1}=0$ for consistency. 

As before, we impose the conditions for a highest-weight representation
\begin{align}
H \tilde{\theta}^i &= -(2i-p) \tilde{\theta}^i ~,&
H \tilde{\eta}^i &= -(2i-p) \tilde{\eta}^i ~,\\
H \eta^i &= -(p-2i)\eta^i ~,&
H \theta^i & = - (p- 2i) \theta^i ~.
\end{align}
The solution is unique up to signs and we find
\begin{align*}
  a_i^2 &= b_i^2= (p-i)(i+1) \, , \quad i = 0,\cdots,\frac{p-1}2-1\\
  c^2&= \left( \frac{p+1}2 \right)^2~.
\end{align*}
The (hyperbolic) element $H$ is
\begin{equation}
H = \sum_{i=0}^{\frac{p-1}2}(p-2i)( \tilde{\theta}^i \wedge \theta^i + \tilde{\eta}^i \wedge \eta^i )~.
\end{equation}
Using this, one can construct the Kostant-Segikuchi triple in $\mathfrak{p}^{\CC}$. Equation \eqref{eq:einpctypef} is given by switching to an orthonormal frame
\begin{align}
e^{(1)}_i &= \frac{\sqrt{2}}{{2}} \left( \eta^i + \tilde{\eta}^i \right) \\
\hat{e}^{(1)}_i &= \frac{\sqrt{2}}{{2}} \left( \eta^i - \tilde{\eta}^i \right) \\
e^{(2)}_i &= \frac{\sqrt{2}}{{2}} \left( \theta^i + \tilde{\theta}^i \right) \\
\hat{e}^{(2)}_i &= \frac{\sqrt{2}}{{2}} \left( \theta^i - \tilde{\theta}^i  \right)~.
\end{align}

\subsection{Normal forms in $\mathfrak{p}^{\CC}$}
\label{app:norminpc}

\subsubsection{Type  $\Delta_p^{\pm}(0)$}
Let us first write the indecomposable nilpotent element in  $\mathfrak{p}^{\CC}$ corresponding to type $\Delta_p^{\pm}(0)$, where $p$ is even. We use the
orthonormal basis $\{e_i,\hat{e}_i,\tilde{e}\}$, $i=0,\ldots, \frac{p}2-1$, of $\RR^{\frac{p}2+1,\frac{p}2}$ (respectively of $\RR^{\frac{p}2,\frac{p}2+1}$) 
where $\tilde{e}$ is spacelike (respectively timelike) if $\pm (-1)^{\frac{p}2}$ is $+1$ (respectively $-1$). Then, the following is a normal form for the class
\begin{equation}
\begin{aligned}
e &= \frac12 
\Big( \sum_{i=0}^{\frac{p}2-2} a_i \left(
  \hat{e}_i \wedge e_{i+1} + \hat{e}_{i+1}\wedge e_i \right) \\
&+
\sqrt{2}\, b \, \tilde{e} \wedge 
\left\{
\begin{array}{rl}
-\hat{e}_{\frac{p}2-1} ~, & \text{if }\pm (-1)^{\frac{p}2}=1 \\
e_{\frac{p}2-1} ~,& \text{if }\pm (-1)^{\frac{p}2}=-1
\end{array}
\right\}
\\
&+
i \sum_{i=0}^{\frac{p}2-1}(p-2i) e_i\wedge \hat{e}_{i}
\Big)~,
\end{aligned}
\label{eq:einpctypee}
\end{equation}
where
\begin{align*}
  a_i^2 &= (p-i)(i+1) ~,\quad i=0,\cdots, \frac{p}2-2\\
  b^2 &=\frac{p}2 \left( \frac{p}2 +1\right)~.
\end{align*}

\subsubsection{Type $\Delta_p(0,0)$}
We now write the element corresponding to the type $\Delta_p(0,0)$, where $p$ is odd. We use the 
orthonormal 
basis $\{ e_i^{(1)}, e_i^{(2)}, \hat{e}_i^{(1)}, \hat{e}^{(2)}_i \}$, $i=0,\ldots, \frac{p-1}2$, of $\mathbb{R}^{p+1,p+1}$, where the $e^{(j)}_i$ 
are spacelike and the $\hat{e}^{(j)}_i$ are timelike. The nilpotent element is
\begin{equation}
\label{eq:einpctypef}
\begin{aligned}
e &= \frac12 \Big(
\sum_{i=0}^{\frac{p-1}2-1}a_i\left( 
\hat{e}_{i+1}^{(1)}\wedge e_i^{(1)} 
+ \hat{e}_i^{(1)} \wedge e_{i+1}^{(1)}
+ \hat{e}_{i+1}^{(2)}\wedge e_i^{(2)} 
+ \hat{e}_i^{(2)} \wedge e_{i+1}^{(2)}
\right)\\
&+ c \left( e_{\frac{p-1}2}^{(2)} \wedge \hat{e}_{\frac{p-1}2}^{(1)}
+ \hat{e}^{(2)}_{\frac{p-1}2}\wedge e^{(1)}_{\frac{p-1}2} \right)\\
&+i \sum_{i=0}^{\frac{p-1}2}(p-2 i)\left(
e_i^{(1)} \wedge \hat{e}^{(1)}_i + e^{(2)}_i \wedge \hat{e}_i^{(2)} 
\right)
\Big)~,
\end{aligned}
\end{equation}
where
\begin{align*}
  a_i^2 &= (p-i)(i+1) \, , \quad i = 0 , \cdots, \frac{p-1}2-1\\
  c^2&= \left( \frac{p+1}2 \right)^2~.
\end{align*}

\subsection{Proof of theorem \ref{thm:susynil}}
\label{app:susynil}
In this section we prove theorem \ref{thm:susynil} on page \pageref{thm:susynil}. In order to facilitate our calculations, let us use the notation of
Clifford multiplication $v \epsilon$ of a vector $v$ in $Cl(8,0)$
acting on a spinor $\epsilon$ of the Clifford module, and similarly for a
higher-degree form. 

\begin{proof}[Proof of (a) and (b)]
Assume $\Delta_p(0,0)$ appears in the decomposition of $\Pz$ with $p\geq 3$. Let us use the orthonormal 
basis  $\{ e^{(1)}_i, e^{(2)}_i, \hat{e}^{(1)}_i, \hat{e}^{(2)}_i \}$, $i=0,\cdots, \frac{p-1}2$, of $\mathbb{R}^{p+1,p+1}$, where 
$p>3$ is odd and $\mathbb{R}^{p+1,p+1}$ is an orthogonal subspace of 
$\mathbb{R}^{8,n}$. That is, the basis $\{  e^{(1)}_i ,  e^{(2)}_i , \hat{e}^{(1)}_i ,  \hat{e}^{(2)}_i \}$, $i=0,\cdots ,\frac{p-1}2$, 
is a subbasis of some orthonormal basis $\{e_I,\hat{e}_r\}$, $I=1,\ldots ,8$ and $r=1,\dots, n$, of $\RR^{8,n}$. 
According to \eqref{eq:einpctypef}, the nilpotent element is of the form
\begin{equation}
P^{Ir} e_I \otimes \hat{e}_r = e + \cdots
\end{equation}
with
\begin{equation}
e =\frac{1}{2}\left( -  a_0 \, e^{(1)}_1 \otimes \hat{e}^{(1)}_0 + i\, p\,  e^{(1)}_0  \otimes \hat{e}^{(1)}_0 \right) + \cdots
\end{equation}
where in ``$\cdots$'' of both equations, the vector $\hat{e}_0$ does not
appear again. The algebraic supersymmetry equation \eqref{eq:algsusy} for the index $r$ corresponding
to the direction of $\hat{e}^{(1)}_0$ becomes
\begin{equation}\label{eq:bps1}
\left(-a_0 e^{(1)}_1 + i \, p\,  e^{(1)}_0 \right)\epsilon = 0~.
\end{equation}
In this equation, we are assuming the Clifford multiplication of the
vectors $ e^{(1)}_0$ and $ e^{(1)}_1$ in $\mathbb{R}^{p+1}\subset \RR^{8}$ in the Clifford module of
$Cl(8,0)$. Multiplying with $ e^{(1)}_0$ and using $| e^{(1)}_0|=1$ in
$\mathbb{R}^{8,0}$, \eqref{eq:bps1} becomes
\begin{equation}
-a_0 e^{(1)}_0 \wedge  e^{(1)}_1 \epsilon = i \, p \, \epsilon~,
\end{equation}
where again $ e^{(1)}_0 \wedge  e^{(1)}_1\,  \epsilon$ is the Clifford action of the
two-form on the complex spinor. Since $ e^{(1)}_0 \wedge  e^{(1)}_1$ squares to $-1$ in the Clifford algebra $Cl(8,0)$, whereas
\begin{equation}
a_0^2 = p \neq p^2~,
\end{equation}
the only solution is $\epsilon=0$ and there is thus no supersymmetry. On the other hand, the indecomposable complex element of type $\Delta_1(0,0)$ is
\begin{equation}
\label{eq:eisdelta1}
e=\frac12 \left( \pm \left( -
    e^{(2)}_0 \otimes \hat{e}^{(1)}_0  +  e^{(1)}_0 \otimes  \hat{e}^{(2)}_0 \right) + i
  \left(  e^{(1)}_0 \otimes \hat{e}^{(1)}_0 +   e^{(2)}_0 \otimes  \hat{e}^{(2)}_0 \right) \right)~,
\end{equation}
where the $\pm1$ sign is the sign of $a_0$. The algebraic supersymmetry
equation \eqref{eq:algsusy} becomes
\begin{equation}
\left(\pm  e^{(2)}_0 + i \,  e^{(1)}_0 \right)\epsilon =0 ~.
\end{equation}
Indeed, this equation is obtained for $r$ corresponding to either
the direction of $\hat{e}^{(1)}_0$ or $ \hat{e}^{(2)}_0$. This is a BPS-type projection that halves supersymmetry.
\end{proof}

\begin{proof}[Proof of (c)]
Assume $\Delta_p^{\pm}(0)$ appears in the decomposition with $p\geq 4$. 
For simplicity, let us take $\pm(-1)^{\frac{p}2}=+1$, while the
proof is completely analogous for the opposite sign. We use the
orthonormal basis  $\{ e_i , \hat{e}_i , \tilde{e} \}$, $i=0,\cdots, \frac{p}2-1$,
of $\mathbb{R}^{\frac{p}2+1,\frac{p}2}$, where $p$ is even and
$\mathbb{R}^{\frac{p}2+1,\frac{p}2}$ is an orthogonal subspace of
$\mathbb{R}^{8,n}$. That is, the basis  $\{  e_i , \hat{e}_i , \tilde{e}  \}$, $i=0,\cdots, \frac{p}2-1$, is a 
subbasis of some orthonormal basis $\{e_I,\hat{e}_r\}$, $I=1,\ldots ,8$ and $r=1,\dots, n$, of $\RR^{8,n}$. According to \eqref{eq:einpctypef}, the nilpotent element 
is of the form
\begin{equation}
\label{eq:p1}
P^{Ir}_z e_I \otimes \hat{e}_r = e + \cdots
\end{equation}
with
\begin{equation}
\label{eq:e1}
e =\frac{1}{2}\left( -  a_0 \, e_1 \otimes \hat{e}_0  + i\, p\,  e_i  \hat{e}_0 \otimes  \hat{e}_0 \right) + \cdots
\end{equation}
where in ``$\cdots$'' of both equations, the vector $\hat{e}_0$ does not
appear again. The proof here then proceeds similarly to the proof of (a). The algebraic supersymmetry equation \eqref{eq:algsusy} for the index $r$ corresponding
to the direction of $\hat{e}_0$ becomes
\begin{equation}\label{eq:bps2}
\left(-a_0  e_1 + i \, p\, e_0\right)\epsilon = 0~.
\end{equation}
where again $e_1$ and  $e_0$ square to $-1$, while $
a_0^2 = p \neq p^2~.
$ There is thus no supersymmetry allowed. 
\end{proof}
\begin{proof}[Proof of (d)]
Type $\Delta_0(0)^{\pm}$
corresponding to a spacelike or timelike $\mathbb{R} \subset
\mathbb{R}^{8,n}$ is such that $e = 0$. It imposes no supersymmetry
restriction itself from the algebraic supersymmetry equation \eqref{eq:algsusy}.
\end{proof}
\begin{proof}[Proof of (e)]
Take now $p=2$ and consider $\Delta_2^+(0)$. We assume as before a
basis $\{ e_0 , \hat{e}_0 , \tilde{e}\}$ of $\RR^{1,2}$. The nilpotent element $\Pz$
is again of the form
\begin{equation}
P^{Ir}_z e_I \otimes \hat{e}_r = e + \cdots
\end{equation}
with
\begin{equation}
\label{eq:eiscasee}
e = 
\left( \pm \, e_0 \otimes \tilde{e} + i e_0
  \otimes \hat{e}_0 \right)~.
\end{equation}
The sign here is that of $b$. 
If we choose the direction of $r$ corresponding to the timelike
$\tilde{e}$, we arrive at the equation
\begin{equation}
e_0 \epsilon = 0~,
\end{equation}
with solution $\epsilon=0$. If we consider $\Delta_2^{-}(0)$ instead
and use the orthonormal basis $\{e_0,\tilde{e},\hat{e}_0\}$ of $\RR^{2,1}$,
\eqref{eq:eiscasee} is replaced by
\begin{equation}
\label{eq:eissusy1}
e = 
\left( \mp \tilde{e} \otimes \hat{e}_0 + i e_0
  \otimes \hat{e}_0 \right)~.
\end{equation}
The sign in this equation is again that of $b$. This is a single
projection equation of the form
\begin{equation}
 i \, e_0 \wedge \tilde{e} \epsilon = \pm \epsilon~.
\end{equation}
Each appearance of $\Delta^{-}_2(0)$ implies a single projection
equation that halves supersymmetry.
\end{proof}

\bibliographystyle{JHEPmod}
\bibliography{ungauged}
\end{document}